\documentclass[11pt]{article}

\usepackage[utf8]{inputenc}
\usepackage{microtype}
\usepackage{amsmath}
\usepackage{amssymb}
\usepackage{amsthm}
\usepackage{bm}
\usepackage{dsfont}
\usepackage{authblk}
\usepackage{fullpage}
\usepackage{complexity}
\usepackage{wrapfig}
\usepackage{comment}
\usepackage[font=small,labelfont=bf]{caption}
\usepackage{nicefrac}
\usepackage{tikz}
\usetikzlibrary{positioning}
\usepackage{mathtools}
\usepackage{bbm}
\usepackage{float}
\usepackage{array,booktabs}
\usepackage[
    backend=biber,
    style=alphabetic,
    maxcitenames=10,
    maxbibnames=99,
    natbib=true,
    url=false,
    doi=false, 
    isbn=false]{biblatex}
\usepackage{hyperref}
\usepackage{mleftright}
\usepackage{thm-restate}
\usepackage[capitalize,noabbrev]{cleveref}
\usepackage[shortlabels]{enumitem}

\newcounter{qst}
\crefname{qst}{Question}{Questions}

\usepackage[ruled,vlined,linesnumbered,noend]{algorithm2e}
\makeatletter
\patchcmd\algocf@Vline{\vrule}{\vrule \kern-0.4pt}{}{}
\patchcmd\algocf@Vsline{\vrule}{\vrule \kern-0.4pt}{}{}
\makeatother

\SetKwComment{Hline}{}{\vspace{-3mm}\textcolor{gray}{\hrule}\vspace{1mm}}
\SetKwComment{HlineWd}{}{\vspace{-3mm}\textcolor{gray}{\hrule width 7.5cm}\vspace{1mm}}
\definecolor{darkgrey}{gray}{0.3}
\definecolor{commentcolor}{gray}{0.5}
\SetKwComment{Comment}{\color{commentcolor}[$\triangleright$\ }{}
\SetCommentSty{}
\SetNlSty{}{\color{darkgrey}}{}
\SetKwProg{Fn}{function}{}{}
\SetKwProg{Subr}{subroutine}{}{}
\crefalias{AlgoLine}{line}%
\crefname{algocf}{Algorithm}{Algorithms}

\makeatletter
\let\cref@old@stepcounter\stepcounter
\def\stepcounter#1{%
  \cref@old@stepcounter{#1}%
  \cref@constructprefix{#1}{\cref@result}%
  \@ifundefined{cref@#1@alias}%
    {\def\@tempa{#1}}%
    {\def\@tempa{\csname cref@#1@alias\endcsname}}%
  \protected@edef\cref@currentlabel{%
    [\@tempa][\arabic{#1}][\cref@result]%
    \csname p@#1\endcsname\csname the#1\endcsname}}
\makeatother

\newcommand{\declarecolor}[2]{\definecolor{#1}{RGB}{#2}\expandafter\newcommand\csname #1\endcsname[1]{\textcolor{#1}{##1}}}
\declarecolor{White}{255, 255, 255}
\declarecolor{Black}{0, 0, 0}
\declarecolor{Maroon}{128, 0, 0}
\declarecolor{Coral}{255, 127, 80}
\declarecolor{Red}{182, 21, 21}
\declarecolor{LimeGreen}{50, 205, 50}
\declarecolor{DarkGreen}{0, 80, 0}
\declarecolor{Navy}{0, 0, 108}
\hypersetup{
	colorlinks=true,
	pdfpagemode=UseNone,
	citecolor=DarkGreen,
	linkcolor=Navy,
	urlcolor=Navy,
	pdfstartview=FitW
}

\theoremstyle{plain}
\newtheorem{theorem}{Theorem}[section]

\newtheorem{corollary}[theorem]{Corollary}
\newtheorem{proposition}[theorem]{Proposition}

\theoremstyle{definition}
\newtheorem{definition}[theorem]{Definition}

\theoremstyle{remark}
\newtheorem{remark}[theorem]{Remark}

\let\E\relax

\newcommand*{\E}{{\mathbb{E}}}
\newcommand*{\N}{{\mathbb{N}}}

\let\R\relax
\newcommand*{\R}{{\mathbb{R}}}

\newcommand*{\cX}{{\mathcal{X}}}
\newcommand*{\cY}{{\mathcal{Y}}}

\newcommand*{\cR}{{\mathcal{R}}}
\newcommand*{\cA}{{\mathcal{A}}}

\newcommand{\defeq}{\coloneqq}

\let\poly\relax
\let\polylog\relax

\DeclareMathOperator{\poly}{poly}
\DeclareMathOperator{\reg}{Reg}
\DeclareMathOperator{\proj}{\Pi}

\DeclareMathOperator{\op}{\mathrm{op}}
\DeclareMathOperator{\polylog}{polylog}

\DeclarePairedDelimiterX{\infdivx}[2]{(}{)}{%
  #1\;\delimsize\|\;#2%
}

\newcommand{\SW}{\textsc{SW}}
\newcommand{\diamset}[1]{\Omega_{#1}}
\newcommand{\diamreg}[1]{\Omega_{#1}}
\newcommand{\normset}[1]{\|#1\|}

\newcommand{\ppad}{\ComplexityFont{PPAD}}

\newcommand{\hatx}{\widehat{\vec{x}}}
\newcommand{\haty}{\widehat{\vec{y}}}

\NewDocumentCommand{\treeset}{o}{\mathbb{T}\IfNoValueF{#1}{_{#1}}}

\newcommand{\rvu}{\texttt{RVU}}

\makeatletter
\AtBeginDocument{%
  \@ifpackageloaded{amsmath}{%
    \newcommand*\patchAmsMathEnvironmentForLineno[1]{%
      \expandafter\let\csname old#1\expandafter\endcsname\csname #1\endcsname
      \expandafter\let\csname oldend#1\expandafter\endcsname\csname end#1\endcsname
      \renewenvironment{#1}%
                       {\linenomath\csname old#1\endcsname}%
                       {\csname oldend#1\endcsname\endlinenomath}%
    }%
    \newcommand*\patchBothAmsMathEnvironmentsForLineno[1]{%
      \patchAmsMathEnvironmentForLineno{#1}%
      \patchAmsMathEnvironmentForLineno{#1*}%
    }%
    \patchBothAmsMathEnvironmentsForLineno{equation}%
    \patchBothAmsMathEnvironmentsForLineno{align}%
    \patchBothAmsMathEnvironmentsForLineno{flalign}%
    \patchBothAmsMathEnvironmentsForLineno{alignat}%
    \patchBothAmsMathEnvironmentsForLineno{gather}%
    \patchBothAmsMathEnvironmentsForLineno{multline}%
  }{}
}
\makeatother

\renewcommand{\vec}[1]{\bm{#1}}
\newcommand{\mat}[1]{\mathbf{#1}}

\newcommand{\range}[1]{[\![#1]\!]}

\DeclarePairedDelimiterX{\card}[1]{\lvert}{\rvert}{#1}
\DeclarePairedDelimiterX{\abs}[1]{\lvert}{\rvert}{#1}
\DeclarePairedDelimiterX{\norm}[1]{\lVert}{\rVert}{#1}
\DeclarePairedDelimiterX{\tuple}[1]{\lparen}{\rparen}{#1}
\DeclarePairedDelimiterX{\parens}[1]{\lparen}{\rparen}{#1}
\DeclarePairedDelimiterX{\brackets}[1]{\lbrack}{\rbrack}{#1}
\DeclarePairedDelimiterX{\set}[1]\{\}{#1}
\let\Pr\relax
\DeclarePairedDelimiterXPP{\Pr}[1]{\mathbb{P}}[]{}{#1}
\DeclarePairedDelimiterXPP{\PrX}[2]{\mathbb{P}_{#1}}[]{}{#2}
\DeclarePairedDelimiterXPP{\Ex}[1]{\mathbb{E}}[]{}{#1}
\DeclarePairedDelimiterXPP{\ExX}[2]{\mathbb{E}_{#1}}[]{}{#2}

\usepackage{lineno}

\usetikzlibrary{fit,shapes.misc,arrows.meta}
\makeatletter
\tikzset{
  fitting node/.style={
    inner sep=0pt,
    fill=none,
    draw=none,
    reset transform,
    fit={(\pgf@pathminx,\pgf@pathminy) (\pgf@pathmaxx,\pgf@pathmaxy)}
  },
  reset transform/.code={\pgftransformreset}
}
\tikzset{cross/.style={path picture={
  \draw[black]
(path picture bounding box.south east) -- (path picture bounding box.north west) (path picture bounding box.south west) -- (path picture bounding box.north east);
}}}
\tikzstyle{ox}=[semithick,draw=black,circle,cross,inner sep=1.2mm]
\makeatother

\newcommand{\nc}{\newcommand}
\newcount\Comments
\nc\noah[1]{\ifnum\Comments=1 {\textcolor{purple}{[ng: #1]}}\fi}
\nc\maxfish[1]{\ifnum\Comments=1{\textcolor{blue}{[mf: #1]}}\fi}
\nc\costis[1]{\ifnum\Comments=1{\textcolor{brown}{[cd: #1]}}\fi}
\nc\costiss[1]{\textcolor{red}{#1}}

\nc\io[1]{\ifnum\Comments=1 {\textcolor{purple}{[ioannis: #1]}}\fi}

\nc{\Opthedge}{OMWU\xspace}

\nc{\DMO}{\DeclareMathOperator}
\nc\old[1]{\textcolor{brown}{[old: #1]}}
\nc{\BR}{\mathbb{R}}
\nc{\BC}{\mathbb{C}}
\DMO{\Bin}{Bin}
\nc{\BN}{\mathbb{N}}
\nc{\distrs}[1]{\Delta({#1})}
\nc{\BZ}{\mathbb{Z}}
\nc{\ep}{\epsilon}
\nc{\ra}{\rightarrow}
\nc{\st}{\star}
\nc{\Reg}[2]{\REG_{{#1},{#2}}}
\nc{\til}{\tilde}
\nc{\kld}[2]{\KL({#1};{#2})}
\nc{\chisq}[2]{\chi^2({#1};{#2})}
\DMO{\POLYLOG}{polylog}

\nc{\matx}[1]{\left(\begin{matrix}#1\end{matrix}\right)}
\DMO{\VAR}{Var}
\DMO{\COV}{Cov}
\nc{\Var}[2]{\VAR_{{#1}}\left({#2}\right)}
\nc{\Cov}[3]{\COV_{{#1}}\left({#2},{#3}\right)}
\DMO{\DD}{D}
\nc{\fd}[2]{\DD_{#1}{#2}}
\nc{\fds}[3]{\left(\fd{#1}{#2}\right)\^{#3}}
\nc{\fdc}[2]{\DD^\circ_{#1}{#2}}
\nc{\fdcs}[3]{\left(\fdc{#1}{#2}\right)\^{#3}}
\nc{\shf}[2]{\EEE_{#1}{#2}}
\nc{\shfs}[3]{\left(\shf{#1}{#2}\right)\^{#3}}
\nc{\normst}[2]{\left\| {#2} \right\|_{#1}^\st}
\renewcommand{\^}[1]{^{(#1)}}
\DeclareMathOperator*{\argmax}{arg\,max}
\DeclareMathOperator*{\argmin}{arg\,min}

\nc{\grad}{\nabla}
\nc{\lng}{\langle}
\nc{\rng}{\rangle}
\nc{\bbone}{\mathbf{1}}
\nc{\bbzero}{\mathbf{0}}
\nc{\MD}{\mathcal{D}}
\nc{\MM}{\mathcal{M}}
\nc{\MZ}{\mathcal{Z}}
\nc{\MU}{\mathcal{U}}
\nc{\MC}{\mathcal{C}}
\nc{\MT}{\mathbb{T}^{n}}
\nc{\MS}{\mathcal{S}}
\nc{\MX}{\mathcal{X}}
\nc{\MY}{\mathcal{Y}}
\nc{\MB}{\mathcal{B}}
\nc{\MJ}{\mathcal{J}}
\nc{\MF}{\mathcal{F}}
\nc{\MG}{\mathcal{G}}
\nc{\MR}{\mathcal{R}}
\nc{\ML}{\mathcal{L}}
\nc{\MQ}{\mathcal{Q}}

\nc{\ba}{\mathbf{A}}
\nc{\bx}{\mathbf{x}}
\nc{\by}{\mathbf{y}}
\nc{\bz}{\mathbf{z}}
\nc{\bs}{\mathbf{s}}
\nc{\bt}{\mathbf{t}}
\nc{\br}{\mathbf{r}}

\nc{\ME}{\mathcal{E}}
\DMO{\View}{View}
\DMO{\KL}{KL}
\nc{\MW}{\mathcal{W}}
\nc{\CS}{\mathscr{S}}
\nc{\CI}{\mathscr{I}}
\nc{\CQ}{\mathscr{Q}}
\nc{\CL}{\mathscr{L}}
\nc{\CM}{\mathscr{M}}
\nc{\CG}{\mathscr{G}}
\nc{\CR}{\mathscr{R}}

\nc{\wh}{\widehat}

\nc{\BM}{BM\xspace}
\nc{\ALG}{\texttt{ALG}}
\nc{\MCT}{{\rm MCT}}

\nc{\matrixLL}{\mat{L}}
\nc{\vectorecks}{\vec{x}}
\nc{\vectorLL}{\vec{\ell}}
\nc{\matrixKYU}{\mat{Q}}
\nc{\rowdot}{\cdot}

\bibliography{main}

\title{Optimistic Mirror Descent Either Converges to Nash or to Strong Coarse Correlated Equilibria in Bimatrix Games}

\author[1]{Ioannis Anagnostides}
\author[2]{Gabriele Farina}
\author[3]{Ioannis Panageas}
\author[4]{Tuomas Sandholm}

\affil[1,2,4]{Carnegie Mellon University}
\affil[3]{University of California Irvine}
\affil[4]{Strategy Robot, Inc.}
\affil[4]{Optimized Markets, Inc.}
\affil[4]{Strategic Machine, Inc.}
\affil[ ]{\texttt {\{ianagnos,gfarina,sandholm\}@cs.cmu.edu}, and \texttt{ipanagea@ics.uci.edu}}

\date{}

\begin{document}

\maketitle

\pagenumbering{gobble}

\begin{abstract}
    We show that, for any sufficiently small fixed $\epsilon > 0$, when both players in a \emph{general-sum two-player (bimatrix) game} employ \emph{optimistic mirror descent (OMD)} with smooth regularization, learning rate $\eta = O(\epsilon^2)$ and $T = \Omega(\poly(1/\epsilon))$ repetitions, either the dynamics reach an \emph{$\epsilon$-approximate Nash equilibrium (NE)}, or the average correlated distribution of play is an \emph{$\Omega(\poly(\epsilon))$-strong coarse correlated equilibrium (CCE)}: any possible unilateral deviation does not only leave the player worse, but will decrease its utility by $\Omega(\poly(\epsilon))$. As an immediate consequence, when the iterates of OMD are bounded away from being Nash equilibria in a bimatrix game, we guarantee convergence to an \emph{exact} CCE after only $O(1)$ iterations. Our results reveal that uncoupled no-regret learning algorithms can converge to CCE in general-sum games remarkably faster than to NE in, for example, zero-sum games. To establish this, we show that when OMD does not reach arbitrarily close to a NE, the \emph{(cumulative) regret} of \emph{both} players is not only \emph{negative}, but \emph{decays linearly} with time. Given that regret is the canonical measure of performance in online learning, our results suggest that cycling behavior of no-regret learning algorithms in games can be justified in terms of efficiency.
\end{abstract}

\clearpage
\pagenumbering{arabic}

\section{Introduction}

The last few years have seen a tremendous amount of progress in computational game solving, witnessed over a series of breakthrough results on benchmark applications in AI~\citep{Bowling15:Heads,Brown17:Superhuman,Moravvcik17:DeepStack,Silver16:Mastering,Vinyals19:Grandmaster}. Most of these advances rely on algorithms for approximating a \emph{Nash equilibrium (NE)}~\citep{Nash50:Equilibrium} in \emph{two-player zero-sum games}. Indeed, in that regime it is by now well-understood how to compute a NE at scale. However, many real-world interactions are not zero-sum, but instead have \emph{general-sum utilities}, and often more than two players. In such settings, Nash equilibria suffer from several drawbacks. First, finding even an approximate NE is \emph{computationally intractable}~\citep{Daskalakis08:Complexity,Etessami07:Complexity,Chen09:Settling,Rubinstein16:Settling}---subject to well-believed complexity-theoretic assumptions. Furthermore, even if we were to reach one, NE outcomes can be dramatically more inefficient in terms of the \emph{social welfare} compared to other more permissive equilibrium concepts~\citep{Moulin78:Strategically}. Finally, NE suffer from \emph{equilibrium selection} issues: there can be a multitude of equilibria, and an equilibrium strategy may perform poorly against the ``wrong'' equilibrium strategies~\citep{Harsanyi88:General,Harsanyi95:A,Matsui95:Approach}, thereby necessitating some form of communication between the players.

A competing notion of rationality is Aumann's concept of \emph{correlated equilibrium (CE)}~\citep{Aumann74:Subjectivity}, generalizing Nash's original concept. Unlike NE, a correlated equilibrium can be computed \emph{exactly} in polynomial time~\citep{Papadimitriou08:Computing,Jiang15:Polynomial}. Further, CE can arise from simple \emph{uncoupled learning dynamics}, overcoming the often unreasonable assumption that players have perfect knowledge over the game's utilities. As such, correlated equilibria constitute a much more plausible outcome under independent agents with \emph{bounded rationality}. Indeed, it is folklore that when all players in a general game employ a \emph{no-regret learning algorithm}~\citep{Hart00:Simple}, the average correlated distribution of play converges to a \emph{coarse correlated equilibrium (CCE)}---a further relaxation of CE~\citep{Moulin78:Strategically}. 

Now as it happens, there are specific classes of games, such as \emph{strictly competitive games}~\citep{Adler09:A} and \emph{constant-sum polymatrix games}~\citep{Daskalakis09:On,Cai11:On,Cai16:Zero}, for which CCE ``collapse'' to NE. In fact, although computing Nash equilibria in such games is amenable to linear programming~\citep{Adler13:The}, the state of the art algorithms are based on uncoupled learning procedures, for reasons mostly relating to scalability. Nevertheless, such algorithms require $\Omega(\poly(1/\epsilon))$ iterations to reach an $\epsilon$-approximate Nash equilibrium~\citep{Daskalakis11:Near}, meaning that the convergence is slow, particularly in the high precision regime---both in theory and in practice. Furthermore, prior literature treats convergence to NE in, for example, zero-sum games analogously to convergence to CCE in general-sum games. While this unifying treatment---which is an artifact of the no-regret framework---may seem compelling at first glance, it is unclear whether CCE share similar convergence properties to NE. Indeed, as we have alluded to, those equilibrium concepts are fundamentally different (in general games). 

Our primary contribution is to challenge this traditional framework, establishing that \emph{convergence to CCE can be remarkably faster than the convergence to NE}; this fundamental difference is showcased in \Cref{fig:NE-CCE}. In fact, we show that the only obstacle for converging \emph{exactly} after only $O(1)$ iterations to a CCE in general-sum two-player games is reaching arbitrarily close to a NE. Our results also reveal an intriguing complementarity: the farther the dynamics are from NE, the faster the guarantee of convergence to CCE. 

\begin{figure}
    \centering
    \includegraphics[scale=.8]{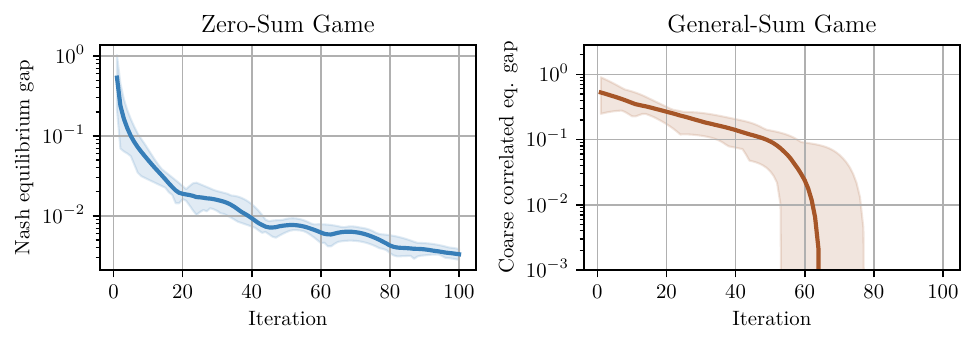}
    \caption{Convergence of uncoupled no-regret dynamics to NE in a zero-sum game (left image) versus convergence to CCE in a general-sum game (right image); further details are provided in \Cref{subsection:example}.}
    \label{fig:NE-CCE}
\end{figure}

\subsection{Our Contributions}

We study uncoupled no-regret learning dynamics in the fundamental class of two-player general-sum games (a.k.a. \emph{bimatrix games}). Specifically, we focus on \emph{optimistic mirror descent (OMD)}~\citep{Chiang12:Online,Rakhlin13:Optimization}, a variant of the standard miror descent method which incorporates a \emph{prediction} into the optimization step. Before we state our main result, let us introduce \emph{strong CCE}, a \emph{refinement} of CCE for which any unilateral deviation from a player is not only worse, but \emph{decreases} its utility by an additive $\epsilon \ge 0$ (see \Cref{def:strong-CCE} for a formal description). In particular, we clarify that any strong CCE is (trivially) an \emph{exact} CCE.\footnote{We include a numerical example illuminating aspects of strong CCE in \Cref{subsection:example}.} We are now ready to state our main theorem.\footnote{For simplicity in the exposition of our results, we use the $O(\cdot)$ and $\Omega(\cdot)$ notation in the main body to suppress parameters that depend on the underlying game; precise statements are given in \Cref{appendix:proofs}.}

\begin{theorem}[Abridged; Full Version in \Cref{corollary:NFGs}]
    \label{theorem:main-abridged}
Fix any sufficiently small $\epsilon > 0$, and suppose that both players in a bimatrix game employ OMD with learning rate $\eta = O(\epsilon^2)$ and smooth regularizer for $T = \Omega(\poly(1/\epsilon))$ repetitions. Then, 
\begin{itemize}
    \item Either the dynamics reach an $\epsilon$-approximate Nash equilibrium;
    \item Or, otherwise, the average correlated distribution of play is an $\Omega(\poly(\epsilon))$-strong CCE.
\end{itemize}
\end{theorem}

Here, the convergence to Nash equilibrium is implied in a \emph{last-iterate} sense (and not a typical time-average): \emph{i.e.}, there exists a time $t \in \range{T}$ for which the dynamics reach an $\epsilon$-approximate Nash equilibrium. Of course, in light of well-established impossibility results, OMD is certainly not going to reach an $\epsilon$-approximate NE---for a sufficiently small \emph{constant} $\epsilon > 0$~\citep{Rubinstein16:Settling}---in \emph{every bimatrix game}. As a result, an immediate interesting implication of \Cref{theorem:main-abridged} is that, setting $\epsilon$ to be sufficiently small, OMD either yields the best known polynomial-time approximation for Nash equilibria in bimatrix games, or, otherwise, $O(1)$ iterations suffice to obtain a strong CCE. It is intriguing that the farther the dynamics are from yielding a Nash equilibrium, the \emph{stronger} the CCE predicted by \Cref{theorem:main-abridged}. Indeed, the only impediment for reaching a strong CCE after only $O(1)$ iterations lies in being very close to a Nash equilibrium.\footnote{Closeness to NE is always implied in terms of the best response gap, not the distance to the set of NE.}

To establish \Cref{theorem:main-abridged}, we prove that when the iterates of OMD are bounded away from being Nash equilibria, the regret of \emph{both} players is not only negative, but \emph{decreases linearly} over time (\Cref{theorem:neg-reg}); see \Cref{fig:regret-trajectories} for an illustration. In this way, our approach represents a substantial departure from prior work which endeavored to characterize no-regret learning dynamics from a unifying standpoint for all possible games. Furthermore, besides the folklore connection of regret with CCE (which allows us to establish \Cref{theorem:main-abridged}), regret is the canonical measure of performance in online learning. In light of this, our results can be construed as follows: when the dynamics do not reach a Nash equilibrium, we obtain remarkably stronger performance guarantees. This seems to suggest that \emph{cycling behavior in games may not be a bug, but a feature}.\footnote{To explain this point further, it is important to connect \emph{last-iterate convergence} with NE. Namely, if all players employ no-regret learning algorithms and the dynamics converge pointwise, then it is easy to see that the limit point has to be a NE. On the other hand, correlated equilibrium concepts are, by definition, incompatible with last-iterate convergence---at least under uncoupled dynamics. Indeed, correlation inherently \emph{requires} cycling behavior. In light of this, there are games for which \emph{stability can be at odds with efficiency} since NE can be dramatically more inefficient that correlation schemes~\citep{Moulin78:Strategically}. We also refer to our example in \Cref{section:experiments} where this phenomenon becomes clear; \emph{c.f.}, see~\citep{Kleinberg11:Beyond}.}

From a technical standpoint, a central ingredient in our proof is a remarkable property---discovered in prior works---coined as the $\rvu$ bound~\citep{Syrgkanis15:Fast,Rakhlin13:Optimization}. Specifically, our argument crucially leverages the rather enigmatic negative term in the $\rvu$ bound (\Cref{def:rvu}), which was treated as a ``cancellation factor'' in prior works. This allows us to establish (conditionally) linear decay for the regrets in \Cref{theorem:neg-reg}. To this end, an important component of our proof is to connect the path lengths of the two players in \Cref{lemma:balanced} through a potential-type argument; notably, such properties break in arbitrary games with more than two players (see \Cref{remark:multi}). (We give a detailed sketch of our proof of \Cref{theorem:main-abridged} in \Cref{section:main}.)

Importantly, our techniques apply under arbitrary convex and compact strategy spaces, thereby allowing for a direct extension to, for example, \emph{normal-form coarse correlated equilibria} in \emph{extensive-form games (EFGs)}~\citep{Moulin78:Strategically}. Another compelling aspect of our result is that both players employ a \emph{constant} learning rate, a feature which has been extensively motivated in prior works (\emph{e.g.}, see \citep{Bailey19:Fast,Piliouras21:Optimal,Golowich20:Tight}). Besides the improved performance guarantees observed in practice under a time-invariant learning rate, it arguably induces a more natural behavior from an economic standpoint. Finally, we corroborate our theoretical findings through experiments conducted on standard benchmark extensive-form games in \Cref{subsection:benchmark}.

\subsection{Further Related Work}

The related work on the subject is too vast to cover in its entirety. Below we only highlight several key contributions, while we encourage the interested reader to investigate further references therein.

\paragraph{Nash Equilibria in Bimatrix Games} Bimatrix games is one of the most fundamental and well-studied classes of games. In a celebrated series of works, it was shown that computing a NE in bimatrix games does not admit a fully polynomial-time approximation scheme (FPTAS), unless every problem in $\ppad$~\citep{Papadimitriou94:On} can be solved in polynomial time~\citep{Chen09:Settling,Daskalakis09:Complexity}. In light of this intractability result, there has been a substantial amount of attention in deriving polynomial-time approximation algorithms for NE in bimatrix games~\citep{Daskalakis06:A,Daskalakis07:Progress,Bosse10:New,Tsaknakis08:An,Kontogiannis10:Well,Daskalakis09:Oblivious,Deligkas22:A,Deligkas22:AWell}. In fact, even computing a sufficiently small \emph{constant} approximation essentially requires quasi-polynomial time~\citep{Rubinstein16:Settling,Kothari18:Sum}, subject to the exponential-time hypothesis for $\ppad$~\citep{Babichenko16:Can}, matching the seminal QPTAS of \citet{Lipton03:Playing}. It is also worth stressing that, in terms of computing Nash equilibria, there is a polynomial-time reduction from any game with a \emph{succinct representation}\footnote{In fact, the reduction in \citep{Daskalakis06:The} also applies to certain \emph{exponential-type} games such as network congestion games and extensive-form games.} to a bimatrix game~\citep{Daskalakis06:The}; this illustrates the generality of bimatrix games.\footnote{An interesting question is whether the reduction in~\citep{Daskalakis06:The} also preserves CCE. In that case, we could essentially lift---in terms of the computational complexity---some of the results established in our paper. Nevertheless, this would not have any immediate practical implications in light of the notoriously complicated reduction developed in~\citep{Daskalakis06:The}.}

\paragraph{Near-Optimal Regret in Games} In a pioneering work by \citet{Daskalakis11:Near}, it was shown that there exist uncoupled no-regret learning dynamics such that when employed by \emph{both} players in a \emph{zero-sum} game, the cumulative regret incurred by each player is bounded by $\widetilde{O}(1)$,\footnote{We use the $\widetilde{O}(\cdot)$ notation to suppress $\polylog(T)$ factors.} substantially improving over the $\Omega(\sqrt{T})$ lower bound under \emph{adversarial utilities}. At the same time, the algorithm proposed in \citep{Daskalakis11:Near} enjoys optimal regret in the adversarial regime as well. Thereafter, there has been a considerable amount of effort in simplifying and extending the previous results along several lines~\citep{Rakhlin13:Optimization,Syrgkanis15:Fast,Chen20:Hedging,Foster16:Learning,Wei18:More,Farina19:Stable,Daskalakis21:Fast,Daskalakis21:Near,Piliouras21:Optimal}. This line of work was culminated in a recent result due to \citet{Daskalakis21:Near}, establishing $\widetilde{O}(1)$ individual regret when all players employ an \emph{optimistic} variant of \emph{multiplicative weights update} in general games, nearly-matching the lower bound in~\citep{Daskalakis11:Near}. Nevertheless, the optimality here is only \emph{existential}: there \emph{exist} games---in fact, zero-sum---for which the guarantee in~\citep{Daskalakis21:Near} is essentially unimprovable. But we argue that this is not a good enough justification for considering the problem of no-regret learning in games settled. As we show in this paper, substantially more refined guarantees are possible beyond zero-sum games.

\paragraph{Last-Iterate Convergence} A folklore phenomenon in the theory of learning in games is that broad families of no-regret algorithms exhibit \emph{cyclic} or even \emph{chaotic behavior}~\citep{Yuzuru02:Chaos,Sandholm10:Population,Mertikopoulos18:Cycles,Andrade21:Learning,Cheung20:Chaos,Cheung21:Chaos}, even in low-dimensional zero-sum games. A compelling approach to ameliorate this problem was proposed by~\citet{Daskalakis18:Training}, showing that an \emph{optimistic} variant of gradient descent guarantees last-iterate convergence in unconstrained bilinear (zero-sum) games. Thereafter, there has been a tremendous amount of interest in strengthening and extending their result in different regimes; for a highly incomplete list, we refer to~\citep{Daskalakis19:Last,Daskalakis18:The,Wei21:Linear,Golowich20:Tight,Golowich20:Last,Hsieh21:Adaptive,Zhou17:Mirror,LinZMJ20,Mertikopoulos19:Optimistic,Liang19:Interaction,Zhang20:Convergence,Mokhtari20:A,Daskalakis20:Independent,Wei21:Last,Anagnostides21:Frequency}, and references therein. The precursor of our main result is the recent paper of~\citet{Anagnostides22:On} that obtained similar results but for the \emph{sum} of the players' regrets, which does not have implications in terms of convergence to CCE. We point out that the last-iterate convergence of instances of OMD in zero-sum games is consistent with---in fact, implied by---\Cref{theorem:main-abridged}. Indeed, in zero-sum games the sum of the players' regrets has to be nonnegative since it is tantamount to the \emph{duality gap}. Thus, \Cref{theorem:neg-reg} implies that OMD has to reach arbitrarily close to a Nash equilibrium in zero-sum games; this argument also extends to \emph{strategically zero-sum games}~\citep{Moulin78:Strategically}.
\section{Preliminaries}
\label{section:prel}

In this section we introduce basic notions related to online optimization and equilibria in games. For a comprehensive treatment on online optimization we refer to~\citep{Shalev-Shwartz12:Online}, while for an introduction to the theory of learning in games we refer to the excellent book of \citet{Cesa-Bianchi06:Prediction}. This section can be skimmed for notation if the reader is already familiar with those topics.

\paragraph{Notation and Conventions} We let $\N \defeq \{1, 2, \dots \}$ be the set of natural numbers. The (discrete) time index is represented exclusively via the variable $t$. We also use the shorthand notation $\range{p} \defeq \{1, 2, \dots, p\}$. We let $\|\cdot\|$ be a norm on $\R^d$, for some $d \in \N$, and $\|\cdot\|_*$ be the \emph{dual norm} of $\|\cdot\|$; namely, $\| \Vec{v} \|_* \defeq \sup_{\|\vec{u}\| \leq 1} \langle \vec{u}, \vec{v} \rangle$. For a (nonempty) convex and compact set $\cX \subseteq \R^d$, we denote by $\diamset{\cX} \defeq \max_{\vec{x}, \hatx \in \cX} \| \vec{x} - \hatx \|$ the \emph{diameter} of $\cX$ with respect to $\|\cdot\|$; to lighten the notation, the underlying norm will be implicit in our notation. Moreover, we let $\normset{\cX} \defeq \max_{\vec{x} \in \cX} \|\vec{x}\|$. For a matrix $\mat{A}$ we let $\|\mat{A}\|_{\op}$ be its \emph{operator norm}: $\|\mat{A}\|_{\op} \defeq \sup_{\|\vec{u}\| \leq 1} \|\mat{A} \vec{u}\|_*$. We point out that $\|\mat{A}\|_{\op} \neq 0$ when $\mat{A} \neq 0$. Finally, for a finite nonempty set $\cA$, we let $\Delta(\cA) \defeq \left\{ \vec{x} \in \R_{\ge 0}^{|\cA|} : \sum_{a \in \cA} \vec{x}(a) = 1 \right\}$ be the \emph{probability simplex} over $\cA$, where $\vec{x}(a)$ is the coordinate of $\vec{x}$ corresponding to $a \in \cA$.

\subsection{Online Learning and Regret}

Let $\cX$ be a nonempty, convex and compact subset of $\R^d$, for some $d \in \N$. In the \emph{online learning framework} the \emph{learner} commits to a \emph{strategy} $\Vec{x}^{(t)} \in \cX$ at every iteration $t \in \N$. Then, the learner receives as feedback from the environment a (linear) \emph{utility function} $u^{(t)} : \cX \ni \Vec{x} \mapsto \langle \Vec{x}, \Vec{u}^{(t)} \rangle$, with $\vec{u}^{(t)} \in \R^d$, so that the utility received at time $t$ is given by $\langle \Vec{x}^{(t)}, \Vec{u}^{(t)} \rangle$. Based on that feedback, the learner may choose to adapt its next strategy accordingly. The framework is \emph{online} in the sense that no information about future utilities is available. The canonical objective in this framework is to minimize the \emph{cumulative external regret} (or simply \emph{regret}), defined as follows.
\begin{equation}
    \label{eq:reg}
    \reg^T \defeq \max_{\Vec{x}^* \in \cX} \left\{ \sum_{t=1}^T \langle \Vec{x}^*, \Vec{u}^{(t)} \rangle \right\} - \sum_{t=1}^T \langle \Vec{x}^{(t)}, \Vec{u}^{(t)} \rangle,
\end{equation}
where $T \in \N$ is the \emph{time horizon}. That is, the performance is measured in terms of the optimal \emph{fixed strategy in hindsight}. It is well-known that broad families of online learning algorithms, such as \emph{mirror descent (MD)}~\citep{Nemirovski83:Problem}, ensure $O(\sqrt{T})$ regret under \emph{any sequence of (bounded) utilities}, and this guarantee is asymptotically optimal in the \emph{adversarial regime}~\citep{Cesa-Bianchi06:Prediction}. 

\subsection{Optimistic Regret Minimization} 

On the other hand, much better guarantees are attainable beyond the worst-case setting if the observed utilities are more \emph{predictable}. For example, this could occur if the utilities have small variation. To leverage the additional structure, several \emph{predictive} algorithms have been introduced in the last few years~\citep{Hazan11:Better,Chiang12:Online,Chiang13:Beating,Rakhlin13:Online,Rakhlin13:Optimization,Syrgkanis15:Fast,Farina21:Faster}. In this paper we employ a predictive variant of MD, known as \emph{optimistic mirror descent}~\citep{Chiang12:Online,Rakhlin13:Optimization}.

\paragraph{Optimistic Mirror Descent} 

Let $\cR$ be a $1$-strongly convex regularizer (or \emph{distance generating function (DGF)}) with respect to a norm $\|\cdot\|$ on $\R^d$, continuously differentiable on $\cX$.\footnote{In general, differentiability of $\cR$ is only imposed for the relative interior of $\cX$, but for our purposes we will require this stronger condition. Regularizers whose gradients blow up in the boundary, such as \emph{negative entropy}~\citep{Shalev-Shwartz12:Online}, will not satisfy the smoothness condition we impose in the sequel.} We say that $\cR$ is $G$\emph{-smooth}, with $G > 0$, if for any $\Vec{x}, \hatx \in \cX$,
\begin{equation}
    \label{eq:smoothness}
    \| \nabla \cR(\Vec{x}) - \nabla \cR(\hatx) \|_* \leq G \| \Vec{x} - \hatx \|.
\end{equation}
For example, the \emph{Euclidean regularizer} $\cR(\vec{x}) \defeq \frac{1}{2} \| \vec{x} \|_2^2$, which is $1$-strongly convex with respect to the Euclidean norm $\|\cdot\|_2$, trivially satisfies the smoothness condition of \eqref{eq:smoothness} with $G = 1$; we recall that $\|\cdot\|_2$ is self-dual. Moreover, we let $D_{\cR} \infdivx{\vec{x}}{\hatx} \defeq \cR(\Vec{x}) - \cR(\hatx) - \langle \nabla \cR(\hatx), \Vec{x} - \hatx \rangle$ be the \emph{Bregman divergence} induced by $\cR$~\citep{Rockafellar70:Convex}. \emph{Optimistic mirror descent (OMD)}\footnote{To avoid any confusion, we point out that OMD sometimes stands for \emph{online} mirror descent in the literature.} is parameterized by a (dynamic) \emph{prediction vector} $\Vec{m}^{(t)} \in \R^d$, for every time $t \in \N$, and a \emph{learning rate} $\eta > 0$, so that its update rule takes the following form for $t \in \N$:
\begin{equation}
    \label{eq:OMD}
    \tag{OMD}
    \begin{split}
    \Vec{x}^{(t)} &\defeq \argmax_{\Vec{x} \in \cX} \left\{ \langle \Vec{x}, \Vec{m}^{(t)} \rangle - \frac{1}{\eta} D_{\cR} \infdivx{\Vec{x}}{\hatx^{(t-1)}} \right\}; \\
    \hatx^{(t)} &\defeq \argmax_{\hatx \in \cX} \left\{ \langle \hatx, \Vec{u}^{(t)} \rangle - \frac{1}{\eta} D_{\cR} \infdivx{\hatx}{\hatx^{(t-1)}} \right\}.
    \end{split}
\end{equation}
Further, OMD is initialized as $\hatx^{(0)} \defeq \argmin_{\hatx \in \cX} \cR(\hatx)$; for convenience, we also let $\vec{x}^{(0)} \defeq \hatx^{(0)}$. We will refer to $\hatx^{(0)}, \hatx^{(1)}, \dots$ as the \emph{secondary} sequence of OMD, while $\vec{x}^{(0)}, \vec{x}^{(1)}, \dots$ is the \emph{primary} sequence. We also let $\diamreg{\cR} \defeq \max_{\vec{x} \in \cX} D_{\cR} \infdivx{\vec{x}}{\hatx^{(0)}}$ denote the $\cR$-diameter of $\cX$. As usual, we consider the ``one-recency bias'' prediction mechanism~\citep{Syrgkanis15:Fast}, wherein $\vec{m}^{(t)} \defeq \vec{u}^{(t-1)}$ for $t \in \N$. (To make $\vec{u}^{(0)}$ well-defined in games, each player receives at the beginning the utility corresponding to the players' strategies at $t = 0$; this is only made for convenience in the analysis.) Our results extend beyond this simple prediction mechanism without qualitatively altering our results.

An important special case of \eqref{eq:OMD} arises under the Euclidean regularizer $\cR(\vec{x}) = \frac{1}{2} \|\vec{x}\|_2^2$, in which case we refer to the dynamics as \emph{optimistic gradient descent (OGD)}:
\begin{equation}
    \label{eq:OGD}
    \tag{OGD}
    \begin{split}
    \vec{x}^{(t)} &\defeq \proj_{\cX} \left( \hatx^{(t-1)} + \eta \vec{m}^{(t)} \right); \\
    \hatx^{(t)} &\defeq \proj_{\cX} \left( \hatx^{(t-1)} + \eta \vec{u}^{(t)} \right).    
    \end{split}
\end{equation}
Here, $\proj_{\cX}(\cdot)$ stands for the Euclidean projection to the set $\cX$. We remark that the projection can be computed \emph{exactly} in nearly-linear time for constraint sets such as the simplex (\emph{e.g.}, see \citep{Duchi08:Efficient}).

A fundamental property crystallized in~\citep{Syrgkanis15:Fast}, building on the earlier work of~\citet{Rakhlin13:Optimization}, is the \emph{regret bounded by variation in utilities} ($\rvu$):\footnote{This nomenclature is perhaps unfortunate as it does not capture the crucial last term in \eqref{eq:rvu}.}
\begin{definition}[$\rvu$ Property]
    \label{def:rvu}
A regret minimization algorithm satisfies the $\rvu$ property if there exist time-invariant parameters $\alpha, \beta, \gamma > 0$, so that its regret $\reg^T$ up to $T \in \N$ can be bounded as 
\begin{equation}
    \label{eq:rvu}
    \reg^T \leq \alpha + \beta \sum_{t=1}^T \| \vec{u}^{(t)} - \vec{u}^{(t-1)}\|_*^2 - \gamma \sum_{t=1}^T \| \vec{x}^{(t)} -  \vec{x}^{(t-1)} \|^2,
\end{equation}
under any sequence of utility vectors $\vec{u}^{(1)}, \dots, \vec{u}^{(T)}$. Here, $\vec{x}^{(1)}, \dots, \vec{x}^{(T)}$ is the sequence of iterates produced by the regret minimizer, and $(\|\cdot\|, \|\cdot\|_*)$ is a pair of dual norms.
\end{definition}

The first immediate consequence of the $\rvu$ bound is that, as long as the observed utilities exhibit small variation over time, the intermediate term in \eqref{eq:rvu} will also grow slowly, leading to improved regret bounds. For example, when \emph{all} players in a general normal-form game employ regularized learning algorithm satisfying the $\rvu$ bound, then the \emph{individual regret} of each player grows as $O(T^{1/4})$~\citep{Syrgkanis15:Fast}, a substantial improvement over the $\Omega(\sqrt{T})$ lower bound under adversarial utilities. And, as it turns out, online learning algorithms such as OMD satisfy \Cref{def:rvu} with parameters that depend only on the learning rate and the geometry of the regularizer (\Cref{prop:refined_rvu}). But the even more remarkable feature of the $\rvu$ bound---which is crucially leveraged in our work---is the last term in \eqref{eq:rvu}, which seems to suggest that the regret could decay substantially when the \emph{iterates change rapidly over time}. For our purposes, we will employ a slight refinement of the $\rvu$ bound for OMD, which follows from~\citep[Lemma 1 in the Full Version]{Rakhlin13:Optimization}---this can be extracted en route to the proof of~\citep[Theorem 18]{Syrgkanis15:Fast}.
\begin{proposition}[\citep{Rakhlin13:Optimization,Syrgkanis15:Fast}]
    \label{prop:refined_rvu}
The regret of \eqref{eq:OMD} with learning rate $\eta > 0$ can be bounded as 
\begin{equation*}
    \reg^T \leq \frac{\diamreg{\cR}}{\eta} + \eta \sum_{t=1}^T \| \vec{u}^{(t)} - \vec{u}^{(t-1)} \|^2_* - \frac{1}{4\eta} \sum_{t=1}^T \left(  \| \vec{x}^{(t)} - \hatx^{(t)} \|^2 + \| \vec{x}^{(t)} - \hatx^{(t-1)} \|^2 \right).
\end{equation*}
\end{proposition}
For convenience, we will use a shorthand notation for the \emph{second-order path length}:
\begin{equation}
    \label{eq:not}
    \Sigma_{\cX}^T \defeq  \sum_{t=1}^T \left(  \| \vec{x}^{(t)} - \hatx^{(t)} \|^2 + \| \vec{x}^{(t)} - \hatx^{(t-1)} \|^2 \right).
\end{equation}

\subsection{No-Regret Learning and Coarse Correlated Equilibria}

A folklore connection ensures that when all players in a general game employ a no-regret learning algorithm, the average correlated distribution of play converges to a \emph{coarse correlated equilibrium (CCE)}~\citep{Aumann74:Subjectivity,Moulin78:Strategically}. Formally, let us first introduce the notion of a CCE in general \emph{normal-form games (NFGs)}. To this end, we consider a set of $p$ players $\range{p}$, with each player $i$ having a set of available actions $\cA_i$. The \emph{utility} of player $i \in \range{p}$ is a function $u_i : \bigtimes_{i=1}^p \cA_i \ni \vec{a} \mapsto \R$ indicating the utility received by player $i$ under the action profile $\vec{a} = (a_1, \dots, a_p)$.
\begin{definition}[Approximate Coarse Correlated Equilibrium]
    \label{def:CCE}
    A distribution $\vec{\mu}$ over the set $\bigtimes_{i=1}^p \cA_i$ is an \emph{$\epsilon$-approximate coarse correlated equilibrium}, with $\epsilon \geq 0$, if for any player $i \in \range{p}$ and any unilateral deviation $a_i' \in \cA_i$,
    \begin{equation*}
        \E_{\vec{a} \sim \vec{\mu}}\left[ u_i(a_i', \vec{a}_{-i}) \right] \leq \E_{\vec{a} \sim \vec{\mu}}\left[ u_i(\vec{a}) \right]  + \epsilon.
    \end{equation*}
\end{definition}

A CCE is typically modeled via a trusted mediator who privately recommends actions drawn from a commonly known correlated distribution. Specifically, in a CCE following the mediator's suggestion---\emph{before} actually seeing the recommendation---is a best response (in expectation) for all players. We are now ready to state the fundamental theorem connecting no-regret learning with CCE; for completeness, we include the short proof in \Cref{appendix:proofs}.
\begin{restatable}[Folklore]{theorem}{folklore}
    \label{theorem:folklore}
    Suppose that every player $i \in \range{p}$ employs a no-regret learning algorithm with regret $\reg_i^T$ up to time $T \in \N$. Moreover, let $\vec{\mu}^{(t)} \defeq \vec{x}_1^{(t)} \otimes \dots \otimes \vec{x}_p^{(t)}$ be the correlated distribution of play at time $t \in \range{T}$, and $\Bar{\vec{\mu}} \defeq \frac{1}{T} \sum_{t=1}^T \vec{\mu}^{(t)}$ be the average correlated distribution of play up to time $T$. Then, 
    \begin{equation*}
        \E_{\vec{a} \sim \bar{\vec{\mu}}}\left[ u_i(a_i', \vec{a}_{-i}) \right] \leq \E_{\vec{a} \sim \bar{\vec{\mu}}}\left[ u_i(\vec{a}) \right] + \frac{1}{T} \max_{i \in \range{p}}\reg_i^T.
    \end{equation*}
\end{restatable}
For example, when $\reg_i^T = \widetilde{O}(1)$ for all $i \in \range{p}$~\citep{Daskalakis21:Near}, the average correlated distribution of play converges at a rate of $\widetilde{O}(1/T)$ to a CCE. We also introduce a refinement of CCE which we refer to as \emph{strong CCE}:
\begin{definition}[Strong Coarse Correlated Equilibrium]
    \label{def:strong-CCE}
    A probability distribution $\vec{\mu}$ over the set $\bigtimes_{i=1}^p \cA_i$ is an \emph{$\epsilon$-strong coarse correlated equilibrium}, with $\epsilon \geq 0$, if for any player $i \in \range{p}$ and any unilateral deviation $a_i' \in \cA_i$,
    \begin{equation*}
        \E_{\vec{a} \sim \vec{\mu}}\left[ u_i(a_i', \vec{a}_{-i}) \right] \leq \E_{\vec{a} \sim \vec{\mu}}\left[ u_i(\vec{a}) \right]  - \epsilon.
    \end{equation*}
\end{definition}
In a strong CCE any deviation from the mediator's recommendation is not only worse, but can decrease the player's utility by a significant amount. In that sense, a strong CCE can be a much more self-enforcing equilibrium outcome, and, as such, arguably more likely to occur. A strong CCE (\Cref{def:strong-CCE}) can be thought of as a standard CCE (\Cref{def:CCE}) but with a ``negative approximation''. We include a discussion and an illustration of strong CCE in \Cref{subsection:example}.

\subsection{Bimatrix Games}

A \emph{bimatrix game} involves two players. Each player has a set of strategies $\cX \subseteq \R^n$ and $\cY \subseteq \R^m$, respectively. The (expected) \emph{payoffs} of each player under strategies $(\vec{x}, \vec{y}) \in \cX \times \cY$ is given by the bilinear forms $\vec{x}^\top \mat{A} \vec{y}$ and $\vec{x}^\top \mat{B} \vec{y}$, respectively. Here, $\mat{A}, \mat{B} \in \R^{n \times m}$ are the payoff matrices of the game.\footnote{We assume that $\mat{A} \neq \mat{0}$ and $\mat{B} \neq \mat{0}$. In the contrary case \Cref{theorem:main-abridged} follows trivially.} As an example, the special case where $\cX = \Delta(\cA_{\cX})$ and $\cY = \Delta(\cA_{\cY})$ corresponds to normal-form games, but our current formulation captures \emph{extensive-form games} as well. By convention, we will refer to the two players as player $\cX$ and player $\cY$ respectively. Furthermore, the underlying bimatrix game will be referred to as $(\mat{A}, \mat{B})$, without specifying the strategy sets.

\begin{definition}[Approximate Nash Equilibrium]
    \label{def:NE}
A pair of strategies $(\vec{x}^*, \vec{y}^*) \in \cX \times \cY$ is an \emph{$\epsilon$-approximate Nash equilibrium} of $(\mat{A}, \mat{B})$, for $\epsilon \geq 0$, if for any $(\vec{x}, \vec{y}) \in \cX \times \cY$,
\begin{equation}
    \label{eq:Nash}
    \begin{split}
    \vec{x}^\top \mat{A} \vec{y}^* \leq (\vec{x}^*)^\top \mat{A} \vec{y}^* + \epsilon; \\
    (\vec{x}^*)^\top \mat{B} \vec{y} \leq (\vec{x}^*)^\top \mat{B} \vec{y}^* + \epsilon.
    \end{split}
\end{equation}
\end{definition}
That is, in an $\epsilon$-approximate NE no player has more than an additive $\epsilon \ge 0$ incentive (in expectation) to \emph{unilaterally} deviate from the equilibrium strategy. When $\epsilon = 0$, \eqref{eq:Nash} describes an \emph{exact} NE. While the additive approximation of \Cref{def:NE} we consider here is the typical one encountered in the literature, other notions have also attracted attention, such as relative (or multiplicative) approximations~\citep{Daskalakis13:On}. To make \Cref{def:NE} meaningful, some normalization has to be imposed on the utilities. Here, we will assume that $\max_{\vec{y} \in \cY} \| \mat{A} \vec{y} \|_* \leq 1$ and $\max_{\vec{x} \in \cX} \| \mat{B}^\top \vec{x} \|_* \leq 1$.



\section{Main Result}
\label{section:main}

In this section we sketch the main ingredients required for the proof of \Cref{theorem:main-abridged}; all the proofs are deferred to \Cref{appendix:proofs}. Moreover, a detailed version of \Cref{theorem:main-abridged} for normal-form games under Euclidean regularization is given in \Cref{corollary:NFGs}. In the sequel, we assume that players $\cX$ and $\cY$ employ regularizers $\cR_{\cX}$ and $\cR_{\cY}$, respectively, so that the regret of each player enjoys an $\rvu$ bound with respect to the same pair of dual norms $(\|\cdot\|, \|\cdot\|_*)$; it is immediate to extend the subsequent analysis beyond this case. The first step is to cast the refined $\rvu$ bound of \Cref{prop:refined_rvu} for bimatrix games.

\begin{restatable}{corollary}{rvuour}
    \label{cor:rvu-our}
Suppose that both players employ \eqref{eq:OMD} with learning rate $\eta > 0$. Then, 
\begin{equation*}
    \begin{split}
    \reg_{\cX}^T &\leq \frac{\diamreg{\cR_{\cX}}}{\eta} + \eta \|\mat{A}\|^2_{\op} \sum_{t=1}^T \| \vec{y}^{(t)} - \vec{y}^{(t-1)} \|^2 - \frac{1}{4\eta} \sum_{t=1}^T \left( \| \vec{x}^{(t)} - \hatx^{(t)} \|^2 + \| \vec{x}^{(t)} - \hatx^{(t-1)} \|^2 \right); \\
    \reg_{\cY}^T &\leq \frac{\diamreg{\cR_{\cY}}}{\eta} + \eta \|\mat{B}\|^2_{\op} \sum_{t=1}^T \| \vec{x}^{(t)} - \vec{x}^{(t-1)} \|^2 - \frac{1}{4\eta} \sum_{t=1}^T \left( \| \vec{y}^{(t)} - \haty^{(t)} \|^2 + \| \vec{y}^{(t)} - \haty^{(t-1)} \|^2 \right).
    \end{split}
\end{equation*}
\end{restatable}

The next critical step consists of showing that approximate fixed points of \eqref{eq:OMD} under smooth regularization---in the sense of \eqref{eq:smoothness}---correspond to approximate Nash equilibria (recall \Cref{def:NE}) of the underlying bimatrix game, as we formalize below.

\begin{restatable}[Approximate Fixed Points of OMD]{proposition}{fixedp}
    \label{proposition:approx_stat}
    Consider a bimatrix game $(\mat{A}, \mat{B})$, and suppose that both players employ \eqref{eq:OMD} with learning rate $\eta > 0$ and a $G$-smooth regularizer. Then, if $\| \vec{x}^{(t)} - \hatx^{(t-1)} \|, \| \hatx^{(t)} - \vec{x}^{(t)} \| \leq \epsilon \eta$ and $\| \vec{y}^{(t)} - \haty^{(t-1)} \|, \| \haty^{(t)} - \vec{y}^{(t)} \| \leq \epsilon \eta$, the pair $(\vec{x}^{(t)}, \vec{y}^{(t)})$ is a $(2 \epsilon G \max \{ \Omega_{\cX}, \Omega_{\cY} \} + \epsilon \eta )$-approximate Nash equilibrium of $(\mat{A}, \mat{B})$.
\end{restatable}

Indeed, we show that when the iterates of \eqref{eq:OMD} do not change by much (relatively to the learning rate), each player is approximately best responding to the observed utility. Our argument crucially relies on the smoothness of the regularizer; it appears that \Cref{proposition:approx_stat} does not extend for nonsmooth regularizers such as negative entropy (which generates optimistic multiplicative weights update).

The following ingredient is where we rely on the two-player aspect of the underlying game. On a high level, we show that when the strategies of one of the players change fast over time, the other player ought to be ``moving'' rapidly as well.

\begin{restatable}{lemma}{balanced}
    \label{lemma:balanced}
Suppose that both players in a bimatrix game $(\mat{A}, \mat{B})$ employ \eqref{eq:OMD} with learning rate $\eta > 0$. Then, for any $T \in \N$,
\begin{equation*}
    \begin{split}
    \sum_{t=1}^T \| \vec{y}^{(t)} - \vec{y}^{(t-1)} \| \geq \frac{1}{2\eta \normset{\cX} \|\mat{A}\|_{\op}} \sum_{t=1}^T \left( \| \vec{x}^{(t)} - \hatx^{(t-1)} \|^2 + \| \hatx^{(t)} - \vec{x}^{(t)} \|^2 \right) - \frac{2}{\|\mat{A}\|_{\op}} ; \\
    \sum_{t=1}^T \| \vec{x}^{(t)} - \vec{x}^{(t-1)} \| \geq \frac{1}{2\eta \normset{\cY} \|\mat{B}\|_{\op}} \sum_{t=1}^T \left( \| \vec{y}^{(t)} - \haty^{(t-1)} \|^2 + \| \haty^{(t)} - \vec{y}^{(t)} \|^2 \right) - \frac{2}{\|\mat{B}\|_{\op}}.
    \end{split}
\end{equation*}
\end{restatable}

The intuition is that as long as one player is ``moving'' substantially faster than the other player, its utility will be monotonically increasing as the (repeated) game progresses. But this cannot occur for too long---by a potential argument---since the utility of each player is bounded. As a special case of this phenomenon, we point out that if one of the players remains stationary over time, then the other player should eventually converge to a best response. This is in stark contrast to games with more than two players, as we further explain in \Cref{remark:multi}. We are now ready to state the main technical theorem.

\begin{theorem}[Linear Decay of Regret; Full Version in \Cref{theorem:neg-reg-detailed}]
    \label{theorem:neg-reg}
    Suppose that both players in a bimatrix game $(\mat{A}, \mat{B})$ employ \eqref{eq:OMD} with smooth regularizer, learning rate $\eta = O(\epsilon^2)$ and $T = \Omega\left( \frac{1}{\epsilon^4 \eta^2} \right)$, for a sufficiently small fixed $\epsilon > 0$. Then, if the dynamics do not reach an $O(\epsilon)$-approximate NE, then 
    \begin{equation*}
        \max \{ \reg_{\cX}^T, \reg_{\cY}^T \} \leq - \Omega( \epsilon^4 \eta T).
    \end{equation*}
\end{theorem}

By virtue of \Cref{theorem:folklore} and \Cref{proposition:approx_stat}, \Cref{theorem:neg-reg} immediately implies \Cref{theorem:main-abridged}. An illustration of the linear decay of regret in a bimatrix game can be seen in \Cref{fig:regret-trajectories}. Before we sketch the proof of \Cref{theorem:neg-reg}, let us point out the following useful lemma.

\begin{restatable}[Stability of OMD]{lemma}{stability}
    \label{lemma:stability}
    Suppose that both players employ \eqref{eq:OMD} with learning rate $\eta > 0$. Then, for any $t \in \N$,
    \begin{equation*}
        \begin{split}
        \| \vec{x}^{(t)} - \vec{x}^{(t-1)} \| \le 3 \eta; \\    
        \| \vec{y}^{(t)} - \vec{y}^{(t-1)} \| \le 3 \eta .
        \end{split}
    \end{equation*}
\end{restatable}

\begin{proof}[Sketch Proof of \Cref{theorem:neg-reg}]
When the iterates of \eqref{eq:OMD} are $\Omega(\epsilon)$ from being a Nash equilibrium, \Cref{proposition:approx_stat} implies that $ \Sigma_{\cX}^T + \Sigma_{\cY}^T = \Omega(\epsilon^2 \eta^2 T)$, where we used the notation of \eqref{eq:not} for the second-order path lengths. Now suppose that $\Sigma_{\cX}^T \geq \Sigma_{\cY}^T$. Then, using \Cref{cor:rvu-our} we get that $\reg_{\cX}^T = - \Omega(\epsilon^2 \eta T)$. For the regret of player $\cY$, we first use \Cref{lemma:balanced} to obtain that $\Sigma_{\cY}^T = \Omega(\epsilon^4 \eta^2 T)$. Finally, using \Cref{cor:rvu-our} and \Cref{lemma:stability} we can conclude that $\reg_{\cY}^T \leq \frac{\diamreg{\cR_{\cY}}}{\eta} - \Omega(\epsilon^4 \eta T) = - \Omega(\epsilon^4 \eta T)$.
\end{proof}

While the dichotomy of \Cref{theorem:main-abridged} is based on whether only a \emph{single} iterate is an $\epsilon$-approximate Nash equilibrium, our techniques also directly give analogous guarantees dependening on whether \emph{most}---say $99\%$---of the iterates are $\epsilon$-approximate Nash equilibria, as we formalize in \Cref{cor:mostiter}.

\begin{remark}[Extensive-Form Games]
    \Cref{theorem:main-abridged} has direct implications for \emph{normal-form coarse correlated equilibria (NFCCE)}~\citep{Moulin78:Strategically} in extensive-form games using the \emph{sequence-form strategy representation}~\citep{Romanovskii62:Reduction,VonStengel96:Efficient,Koller96:Efficient}.
\end{remark}

\begin{remark}[Multiplayer Games]
    \label{remark:multi}
\Cref{theorem:main-abridged} does not extend to arbitrary games with $p \ge 3$. To see this, consider a $3$-player game for which the utility of player $3$ does not depend on the strategies of the other players. Then, the regret of player $3$ will be strictly positive---as long as the initialization differs from the optimal strategy. So, even if the dynamics are far from Nash equilibria, the CCE gap will always be strictly positive. At a superficial level, this issue occurs in games where the strategic interactions form, in some sense, multiple ``connected components''. Nevertheless, characterizing the multiplayer games for which \Cref{theorem:main-abridged} holds is an interesting question for the future.
\end{remark}
\section{Experiments}
\label{section:experiments}

In this section we provide experiments supporting our theoretical findings. We start by analyzing the CCE and the behavior of \eqref{eq:OGD} in a simple bimatrix NFG in \Cref{subsection:example}, while in \Cref{subsection:benchmark} we experiment with several benchmark games used in the EFG-solving literature.

\subsection{An Illustrative Example}
\label{subsection:example}

First, we study the $3 \times 3$ bimatrix normal-form game $(\mat{A}, \mat{B})$, where
\begin{equation}
    \label{eq:bimatrix}
    \mat{A} \defeq 
    \begin{bmatrix}
    1 & 0 & 0 \\
    -1 & 1 & 0 \\
    0 & 0 & 1
    \end{bmatrix};
    \mat{B} \defeq 
    \begin{bmatrix}
    0 & 1 & 0 \\
    0 & 0 & 1 \\
    1 & 0 & 0
    \end{bmatrix}.
\end{equation}

This game has a unique Nash equilibrium $(\Vec{x}^*, \Vec{y}^*)$ such that $\Vec{x}^* = (\frac{1}{3}, \frac{1}{3}, \frac{1}{3})$ and $\Vec{y}^* = (\frac{1}{4}, \frac{1}{2}, \frac{1}{4})$~\citep{Avis10:Enumeration}. Moreover, $(\Vec{x}^*, \Vec{y}^*)$ secures a social welfare $\SW(\Vec{x}^*, \Vec{y}^*) \defeq (\Vec{x}^*)^\top \mat{A} \Vec{y}^* + (\Vec{x}^*)^\top \mat{B} \Vec{y}^* = \frac{1}{4} + \frac{1}{3} \approx 0.5833$. 
\begin{wrapfigure}{r}{10cm}
    \centering
    \includegraphics[scale=0.5]{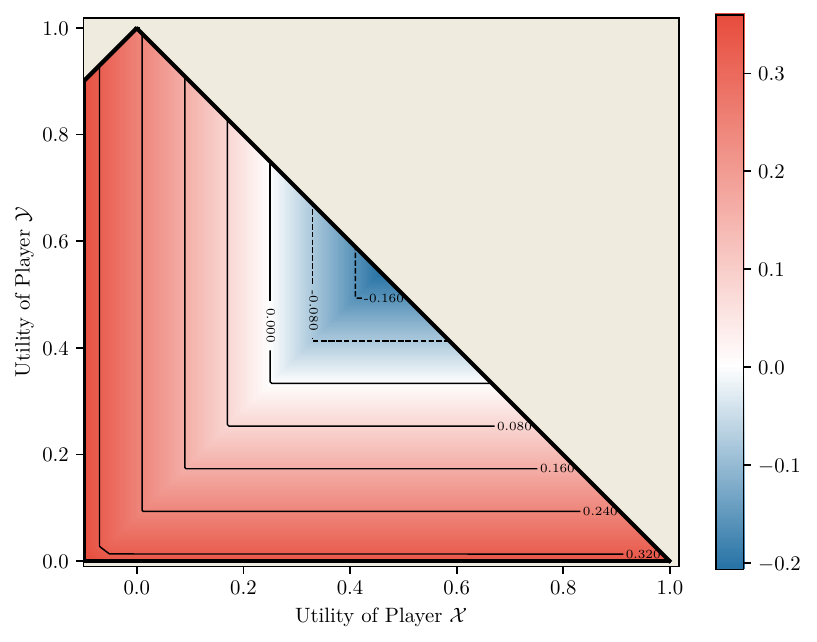}
    \caption{The maximum incentive-compatibility parameter of a CCE which guarantees a given pair of utilities.}
    \label{fig:contours}
\end{wrapfigure}
On the other hand, it is easy to see that there exists an exact CCE $\Vec{\mu}^*$ such that $\SW(\Vec{\mu}^*) = 1$. In fact, this social welfare is optimal even without any incentive-compatibility constraints. Furthermore, using a linear programming solver, we find that the \emph{strongest} CCE (in the sense of \Cref{def:strong-CCE}) has a parameter of roughly $0.2083$. The entire landscape of CCE associated with the bimatrix game \eqref{eq:bimatrix} is illustrated in \Cref{fig:contours}. The blue region corresponds to strong CCE, under which \emph{both} players obtain a high utility. On the other hand, configurations for which one of the players receives low utility are not incentive compatible.

Next, we focus on the behavior of the \eqref{eq:OMD} dynamics. We let both players employ Euclidean regularization and learning rate $\eta \defeq 0.1$. The convergence to CCE of the induced \eqref{eq:OGD} dynamics is illustrated in \Cref{fig:regret-trajectories}. In particular, after $T = 1000$ iterations the average correlated distribution of play $\Bar{\Vec{\mu}}$ reads (with precision up to $4$ decimal places)
\begin{equation*}
    \Bar{\Vec{\mu}} \approx
    \begin{bmatrix}
    0.1594 & 0.1778 & 0.0048 \\
    0.0029 & 0.1614 & 0.1607 \\
    0.1642 & 0.0075 & 0.1613
    \end{bmatrix}.
\end{equation*}
This correlated distribution secures a social welfare of $\SW(\Bar{\Vec{\mu}}) \approx 0.4793 + 0.5027 = 0.9819$. Thus, $\Bar{\Vec{\mu}}$ is near-optimal in terms of the obtained social welfare. As such, it substantially outperforms the efficiency of the Nash equilibrium $(\Vec{x}^*, \Vec{y}^*)$. Moreover, we see that $\Bar{\Vec{\mu}}$ is (approximately) a $0.1525$-strong CCE. Indeed, for player $\cX$ the maximum possible utility attainable from a unilateral deviation is roughly $0.3268$, compared to $0.4793$ obtained under $\Bar{\Vec{\mu}}$; for player $\cY$ the maximum utility from a unilateral deviation is roughly $0.3420$, compared to $0.5027$. It is worth noting that \eqref{eq:OGD} does \emph{not} converge to the strongest possible CCE of the game, even under different random initializations. The results described here are robust to different initializations---although under the definition of \eqref{eq:OGD} each player should start from the uniform distribution. 

Finally, let us elaborate on \Cref{fig:NE-CCE}. The left image illustrates the Nash gap of the average strategies of \eqref{eq:OGD} (with $\eta \defeq 0.1$) in the zero-sum game $(\mat{A}, -\mat{A})$, while the right image shows the CCE gap of the average correlated distribution of play in the bimatrix game $(\mat{A}, \mat{B})$, as given in \eqref{eq:bimatrix}.

\begin{figure}[!ht]
    \centering
    \includegraphics[scale=.8]{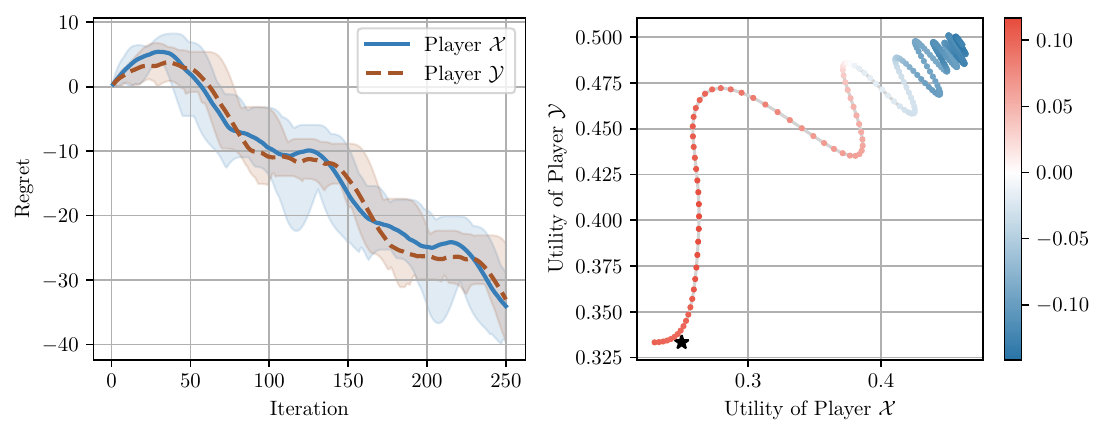}
    \caption{The convergence of \eqref{eq:OGD} dynamics to CCE in the bimatrix game $(\mat{A}, \mat{B})$ introduced in \eqref{eq:bimatrix}. On the left image we plot the regret of each player under different random initializations; we recall that the maximum average regret is tantamount to the CCE gap. We see that after about $60$ iterations both players experience \emph{negative regret}---implying \emph{exact} convergence to CCE. Furthermore, their regret decays linearly over time; this translates to convergence to strong CCE, as illustrated in the right image. More precisely, the color of each point (in the right image) corresponds to the CCE gap at the given iterate. Even when the initialization is ``close'' to the Nash equilibrium, the dynamics manage to avoid it in search of more efficient outcomes. In fact, the limit CCE of the dynamics yields near-optimal social welfare.}
    \label{fig:regret-trajectories}
\end{figure}

\subsection{Benchmark Games}
\label{subsection:benchmark}

Next, we illustrate the convergence of \eqref{eq:OMD} on several benchmark bimatrix EFGs; namely: (i) \emph{Sheriff}~\citep{Farina19:Correlation}; (ii) \emph{Liar's Dice}~\citep{Lisy15:Online}; (iii) \emph{Battleship}~\citep{Farina19:Correlation}; and (iv) \emph{Goofspiel}~\citep{Ross71:Goofspiel}. A detailed description of the game instances we used for our experiments is included in \Cref{appendix:description}.

We instantiated \eqref{eq:OMD} with Euclidean regularization. After a very mild tuning process, we chose for all games a (time-invariant) learning rate of $\eta = (2 \max\{\|\mat{A}\|_2, \|\mat{B}\|_2\})^{-1}$; here, $\|\cdot\|_2$ stands for the spectral norm of the corresponding matrix. For all games the initialization is chosen so that $\hatx^{(0)} \defeq \argmin_{\hatx \in \cX} \cR_{\cX}(\hatx) $ and $\haty^{(0)} \defeq \argmin_{\haty \in \cX} \cR_{\cY}(\haty)$, with the sole exception of Battleship for which that initialization is (virtually) already a Nash equilibrium. In light of this, for Battleship we initialized the dynamics in some arbitrary deterministic strategies; we stress that the conclusions derived here are robust to different initializations. Our results are summarized in \Cref{fig:games}.

\begin{figure}[!ht]
    \centering
    \includegraphics[scale=.69]{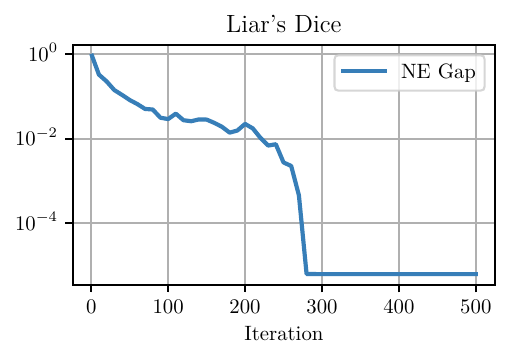}
    \includegraphics[scale=.69]{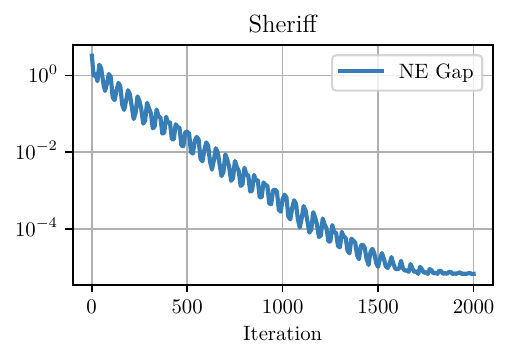}
    \includegraphics[scale=.69]{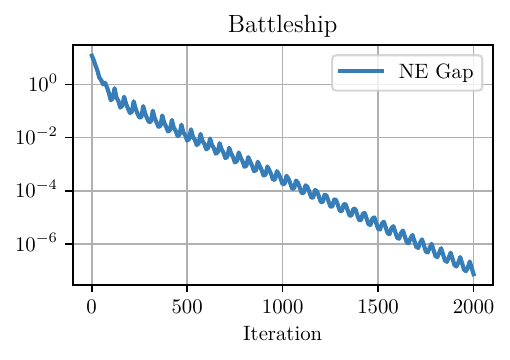}
    \includegraphics[scale=.69]{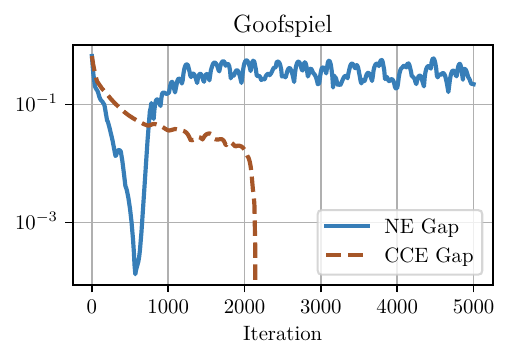}
    \caption{The NE gap of the last iterate and the CCE gap of the average correlated distribution of play under \eqref{eq:OGD} in different benchmark games.}
    \label{fig:games}
\end{figure}

From these benchmark games, only Liar's Dice is constant-sum---in fact, zero-sum. Hence, as expected, the NE gap of the \eqref{eq:OGD} dynamics essentially converges to $0$. Perhaps surprisingly, the same appears to hold for both Sheriff and Battleship. On the other hand, the dynamics exhibit a remarkably different behavior in Goofspiel. Indeed, although initially the dynamics appear to gradually converge to a Nash equilibrium, after about $600$ iterations the NE gap of the last iterate rapidly increases. Afterwards, the CCE gap starts to decay remarkably fast, eventually leading to a strong CCE. Specifically, after $T = 5000$ iterations the average correlated distribution of play is (roughly) a $ 0.0637$-strong CCE. These results are consistent with the predictions of \Cref{theorem:main-abridged}.
\section{Discussion and Open Problems}

Our primary contribution was to establish a new characterization for the convergence properties of OMD---an uncoupled no-regret learning algorithm---when employed by both players in a general-sum game: OMD either reaches arbitrarily close to a Nash equilibrium, or, otherwise, both players experience $- \Omega(T)$ regret. Our results open several interesting avenues for future research. Indeed, below we emphasize on certain key questions.

\begin{itemize}
    \item Can we extend \Cref{theorem:main-abridged} from coarse correlated equilibria to \emph{correlated equilibria}? At the very least, such an extension would be particularly challenging since all the known \emph{no-internal-regret} algorithms used for converging to correlated equilibria involve the stationary distribution of a Markov chain. Even $\rvu$ bounds for no-internal-regret dynamics are not known in the literature. While the recent reduction in~\citep{Anagnostides21:Near} may seem useful for such purposes, their technique only applies for \eqref{eq:OMD} under entropic regularization. In contrast, our current argument crucially relies on the smoothness of the regularizer.
    \item For which classes of multiplayer games would \Cref{theorem:main-abridged} hold? As we pointed out in \Cref{remark:multi}, while \Cref{theorem:main-abridged} does not extend in arbitrary multiplayer games, it is still interesting to give sufficient conditions under which our results would carry over. 
    \item Can we characterize the bimatrix games for which \eqref{eq:OMD} exhibits last-iterate convergence? It should be noted that recent results seem to suggest that such a characterization may be too hard to obtain in general~\citep{Andrade21:Learning}.
    \item Finally, can we improve \Cref{theorem:main-abridged} in terms of the dependence on $T$ and $\eta$? For example, it would be interesting to extend \Cref{theorem:main-abridged} under a learning rate that does not depend on $\epsilon$. Some form of coupling in the spirit of \emph{alternation}~\citep{Tammelin15:Solving} could be useful in that direction.
\end{itemize}

\section*{Acknowledgements}

We are grateful to anonymous NeurIPS reviewers for many helpful comments. Ioannis Panageas is supported by a start-up grant. Part of this work was done while Ioannis Panageas was visiting the Simons Institute for the Theory of Computing. Tuomas Sandholm is supported by the National Science Foundation under grants IIS-1901403 and CCF-1733556.

\printbibliography
\clearpage

\appendix

\section{Omitted Proofs}
\label{appendix:proofs}

In this section we include all the proofs deferred from the main body. For the convenience of the reader, all claims will be restated. Before we proceed with the proof of \Cref{theorem:folklore}, let us point out some additional notational conventions. First, with a standard abuse of notation, we overload
\begin{equation*}
u_i : \bigtimes_{i=1}^p \Delta(\cA_i) \ni (\vec{x}_1, \dots, \vec{x}_p) \mapsto \E_{\vec{a} \sim \vec{x}}[u_i(\vec{a})] = \sum_{\vec{a} \in \cA} u_i(a_1, \dots, a_p) \prod_{j \in \range{p}} \vec{x}_j(a_j)    
\end{equation*}
to denote the mixed extension of player's $i \in \range{p}$ utility function, where $\cA \defeq \bigtimes_{j=1}^p \cA_j$. Furthermore, for $a_i \in \cA_i$ we let
\begin{equation*}
    u_i(a_i, \vec{x}_{-i}) \defeq \sum_{\vec{a}_{-i} \in \cA_{-i}} u_i(a_1, \dots, a_i, \dots, a_p) \prod_{j \neq i} \vec{x}_j(a_j),
\end{equation*}
where $\cA_{-i} \defeq \bigtimes_{j \neq i} \cA_j$.
\folklore*

\begin{proof}
By definition of regret \eqref{eq:reg}, it follows that for any player $i \in \range{p}$ and any possible deviation $a_i' \in \cA_i$,
\begin{equation}
    \label{eq:reg-conseq}
    \sum_{t=1}^T u_i(a_i', \vec{x}_{-i}^{(t)}) - \sum_{t=1}^T u_i(\vec{x}^{(t)}) \leq \reg_i^T.
\end{equation}
Moreover,
\begin{equation*}
    \E_{\vec{a} \sim \bar{\vec{\mu}}}\left[ u_i(\vec{a}) \right] = \frac{1}{T} \sum_{t=1}^T \E_{\vec{a} \sim \vec{\mu}^{(t)}} \left[ u_i(\vec{a}) \right] = \frac{1}{T} \sum_{t=1}^T u_i(\vec{x}^{(t)});
\end{equation*}
and 
\begin{equation*}
    \E_{\vec{a} \sim \bar{\vec{\mu}}}\left[ u_i(a_i', \vec{a}_{-i}) \right] = \frac{1}{T} \sum_{t=1}^T \E_{\vec{a} \sim \vec{\mu}^{(t)}} \left[ u_i(a_i', \vec{a}_{-i}) \right] = \frac{1}{T} \sum_{t=1}^T u_i(a_i', \vec{x}_{-i}^{(t)}).
\end{equation*}
Thus, the theorem follows directly from \eqref{eq:reg-conseq}.
\end{proof}

\rvuour*
\begin{proof}
First, the utility $\vec{u}^{(t)}_{\cX}$ observed by player $\cX$ at time $t \geq 0$ is equal to $\mat{A} \vec{y}^{(t)}$. Thus, the claimed bound for $\reg_{\cX}^T$ follows directly from \Cref{prop:refined_rvu} using the fact that $\| \mat{A} \vec{y}^{(t)} - \mat{A} \vec{y}^{(t-1)} \|_* \leq \|\mat{A}\|_{\op} \| \vec{y}^{(t)} - \vec{y}^{(t-1)}\|$. Similarly, the utility $\vec{u}^{(t)}_{\cY}$ observed by player $\cY$ at time $t \geq 0$ is tantamount to $\mat{B}^\top \vec{x}^{(t)}$. As a result, the bound for $\reg_{\cY}^T$ follows from \Cref{prop:refined_rvu} and the fact that $\|\mat{B}^\top \vec{x}^{(t)} - \mat{B}^\top \vec{x}^{(t-1)}\|_* \leq \| \mat{B}\|_{\op} \| \vec{x}^{(t)} - \vec{x}^{(t-1)} \|$, since $\|\mat{B}^\top\|_{\op} = \|\mat{B}\|_{\op}$.
\end{proof}

\fixedp*

\begin{proof}
First, by definition of the Bregman divergence, the update rule of \eqref{eq:OMD} can be equivalently expressed as 
\begin{equation*}
    \begin{split}
    \Vec{x}^{(t)} &\defeq \argmax_{\Vec{x} \in \cX} \left\{ \langle \Vec{x}, \mathbf{A} \vec{y}^{(t-1)} \rangle - \frac{1}{\eta} \cR(\Vec{x}) + \frac{1}{\eta} \left\langle \Vec{x}, \nabla \cR(\hatx^{(t-1)}) \right\rangle \right\}; \\
    \hatx^{(t)} &\defeq \argmax_{\hatx \in \cX} \left\{ \langle \hatx, \mat{A} \vec{y}^{(t)} \rangle - \frac{1}{\eta} \cR(\hatx) + \frac{1}{\eta} \left\langle \hatx, \nabla \cR(\hatx^{(t-1)}) \right\rangle \right\}.
    \end{split}
\end{equation*}

Now the maximization problem associated with the update rule of the secondary sequence can be equivalently cast in a variational inequality form (\emph{e.g.}, see \citep{Facchinei03:Finite}):
\begin{equation*}
    \left\langle \hatx - \hatx^{(t)}, \mat{A} \vec{y}^{(t)} - \frac{1}{\eta} \left( \nabla \cR(\hatx^{(t)}) - \nabla \cR(\hatx^{(t-1)}) \right) \right\rangle \leq 0, \quad \forall \hatx \in \cX.
\end{equation*}
Thus, 
\begin{align}
    \langle \hatx - \hatx^{(t)}, \mat{A} \vec{y}^{(t)} \rangle &\leq \frac{1}{\eta} \left\langle \hatx - \hatx^{(t)}, \nabla \cR(\hatx^{(t)}) - \nabla \cR(\hatx^{(t-1)}) \right\rangle \notag \\
    &\leq \frac{1}{\eta} \| \hatx - \hatx^{(t)} \| \|\nabla \cR(\hatx^{(t)}) - \nabla \cR(\hatx^{(t-1)})\|_* \label{align:cauchy_schwartz} \\
    &\leq \frac{G}{\eta} \| \hatx - \hatx^{(t)} \| \| \hatx^{(t)} - \hatx^{(t-1)} \| \label{align:smoothness} \\
    &\leq 2 \epsilon G \diamset{\cX}, \label{align:closeness}
\end{align}
for any $\hatx \in \cX$, where \eqref{align:cauchy_schwartz} derives from the Cauchy-Schwarz inequality; \eqref{align:smoothness} follows since $\cR_{\cX}$ is assumed to be $G$-smooth; and \eqref{align:closeness} uses that $\| \hatx^{(t)} - \hatx^{(t-1)} \| \leq 2 \epsilon \eta$, which in turn follows since $\| \hatx^{(t)} - \hatx^{(t-1)} \| \leq \| \hatx^{(t)} - \vec{x}^{(t)} \| + \| \vec{x}^{(t)} - \hatx^{(t-1)} \| \leq 2 \epsilon \eta$ (triangle inequality), as well as the fact that, by definition, $\| \hatx - \hatx^{(t)}\| \leq \diamset{\cX}$ for any $\hatx \in \cX$. As a result, we have shown that for any $\hatx \in \cX$,
\begin{equation}
    \label{eq:approx-br}
    \langle \hatx^{(t)}, \mat{A} \vec{y}^{(t)} \rangle \geq \langle \hatx, \mat{A} \vec{y}^{(t)} \rangle - 2 \epsilon G \Omega_{\cX}.
\end{equation}
Furthermore, by Cauchy-Schwarz inequality we have that 
\begin{equation}
    \label{eq:close-eps}
    \langle \vec{x}^{(t)} - \hatx^{(t)}, \mat{A} \vec{y}^{(t)} \rangle \geq - \| \vec{x}^{(t)} - \hatx^{(t)} \| \mat{A} \vec{y}^{(t)} \|_* \geq - \epsilon \eta,
\end{equation}
where we used the normalization assumption $\| \mat{A} \vec{y}^{(t)}\|_* \leq 1$. Thus, combing \eqref{eq:close-eps} with \eqref{eq:approx-br} yields that 
\begin{equation}
    \label{eq:br-x}
    \langle \vec{x}^{(t)}, \mat{A} \vec{y}^{(t)} \rangle \ge \langle \hatx^{(t)}, \mat{A} \vec{y}^{(t)} \rangle - \epsilon \eta \geq \langle \hatx, \mat{A} \vec{y}^{(t)} \rangle - 2 \epsilon G \diamset{\cX} - \epsilon \eta,
\end{equation}
for any $\hatx \in \cX$. By symmetry, we analogously get that
\begin{equation}
    \label{eq:br-y}
    \langle \vec{y}^{(t)}, \mat{B}^\top \vec{x}^{(t)} \rangle \ge \langle \haty^{(t)}, \mat{B}^\top \vec{x}^{(t)} \rangle - \epsilon \eta \geq \langle \haty, \mat{B}^\top \vec{x}^{(t)} \rangle - 2 \epsilon G \diamset{\cY} - \epsilon \eta,
\end{equation}
for any $\haty \in \cY$. Thus, recalling \Cref{def:NE}, the claim follows from \eqref{eq:br-x} and \eqref{eq:br-y}.
\end{proof}

\balanced*

\begin{proof}
By $1$-strong convexity of $\cR_{\cX}$ with respect to $\|\cdot\|$,
\begin{equation}
    \label{eq:QG-1}
    \begin{split}
        \langle \Vec{x}^{(t)}, \mat{A} \vec{y}^{(t-1)} \rangle - \frac{1}{\eta} D_{\cR_{\cX}} \infdivx{\vec{x}^{(t)}}{\hatx^{(t-1)}} - \langle \hatx^{(t-1)}, \mat{A} \vec{y}^{(t-1)} \rangle \geq \frac{1}{2\eta} \| \Vec{x}^{(t)} - \hatx^{(t-1)} \|^2,
    \end{split}
\end{equation}
where we used the definition of the update rule of the primary sequence of \eqref{eq:OMD}. Similarly,
\begin{equation}
    \label{eq:QG-2}
        \langle \hatx^{(t)}, \mat{A} \vec{y}^{(t)} \rangle - \frac{1}{\eta} D_{\cR_{\cX}} \infdivx{\hatx^{(t)}}{\hatx^{(t-1)}} - \langle \vec{x}^{(t)}, \mat{A} \vec{y}^{(t)} \rangle + \frac{1}{\eta} D_{\cR_{\cX}} \infdivx{\vec{x}^{(t)}}{\hatx^{(t-1)}} \geq \frac{1}{2\eta} \| \hatx^{(t)} - \vec{x}^{(t)} \|^2.
\end{equation}
Hence, summing \eqref{eq:QG-1} and \eqref{eq:QG-2} yields that
\begin{equation*}
    \langle \vec{x}^{(t)}, \mat{A} (\vec{y}^{(t-1)} - \vec{y}^{(t)}) \rangle \geq \frac{1}{2\eta} \left( \| \Vec{x}^{(t)} - \hatx^{(t-1)} \|^2 + \| \hatx^{(t)} - \vec{x}^{(t)} \|^2 \right) - \langle \hatx^{(t)}, \mat{A} \vec{y}^{(t)} \rangle + \langle \hatx^{(t-1)}, \mat{A} \vec{y}^{(t-1)} \rangle,
\end{equation*}
where we used that $D_{\cR_{\cX}} \infdivx{\hatx^{(t)}}{\hatx^{(t-1)}} \geq 0$. Thus, a telescopic summation over all $t \in \range{T}$ implies that
\begin{equation}
    \label{eq:QG}
    \sum_{t=1}^T \langle \vec{x}^{(t)}, \mat{A} (\vec{y}^{(t-1)} - \vec{y}^{(t)}) \rangle \ge \frac{1}{2\eta} \sum_{t=1}^T \left( \| \Vec{x}^{(t)} - \hatx^{(t-1)} \|^2 + \| \hatx^{(t)} - \vec{x}^{(t)} \|^2 \right) - 2 \normset{\cX},
\end{equation}
since $- \langle \hatx^{(T)}, \mat{A} \vec{y}^{(T)} \rangle \ge - \| \hatx^{(T)} \| \|\mat{A} \vec{y}^{(T)} \|_* \geq - \normset{\cX}$ and $ \langle \hatx^{(0)}, \mat{A} \vec{y}^{(0)} \rangle \ge - \| \hatx^{(0)} \| \|\mat{A} \vec{y}^{(0)} \|_* \geq - \normset{\cX}$, where we used the normalization assumption. Furthermore,
\begin{equation*}
    \langle \vec{x}^{(t)}, \mat{A}  ( \vec{y}^{(t-1)} - \vec{y}^{(t)} ) \rangle \leq \| \vec{x}^{(t)} \| \| \mat{A} ( \vec{y}^{(t)} - \vec{y}^{(t-1)} ) \|_* \leq \normset{\cX} \| \mat{A} \|_{\op} \| \vec{y}^{(t)} - \vec{y}^{(t-1)} \|. 
\end{equation*}
Thus, combining this inequality with \eqref{eq:QG} implies that
\begin{equation*}
    \sum_{t=1}^T \| \vec{y}^{(t)} - \vec{y}^{(t-1)} \| \geq \frac{1}{2\eta \normset{\cX} \|\mat{A}\|_{\op}} \sum_{t=1}^T \left( \| \vec{x}^{(t)} - \hatx^{(t-1)} \|^2 + \| \hatx^{(t)} - \vec{x}^{(t)} \|^2 \right) - \frac{2}{\|\mat{A}\|_{\op}}.
\end{equation*}
This completes the first part of the claim. The second part follows analogously by symmetry.
\end{proof}

\stability*

\begin{proof}
Fix any $t \in \N$. By definition of the primary sequence of \eqref{eq:OMD},
\begin{equation*}
    \langle \vec{x}^{(t)}, \mat{A} \vec{y}^{(t-1)} \rangle - \frac{1}{\eta} D_{\cR} \infdivx{\vec{x}^{(t)}}{\hatx^{(t-1)}} - \langle \hatx^{(t-1)}, \mat{A} \vec{y}^{(t-1)} \rangle \geq \frac{1}{2\eta} \| \vec{x}^{(t)} - \hatx^{(t-1)} \|^2,
\end{equation*}
where we used the $1$-strong convexity of the regularizer $\cR_{\cX}$ with respect to $\|\cdot\|$. In turn, this implies that 
\begin{equation*}
    \langle \vec{x}^{(t)} - \hatx^{(t-1)}, \mat{A} \vec{y}^{(t-1)} \rangle \geq \frac{1}{\eta} \| \vec{x}^{(t)} - \hatx^{(t-1)} \|^2,
\end{equation*}
since $D_{\cR_{\cX}} \infdivx{\vec{x}^{(t)}}{\hatx^{(t-1)}} \geq \frac{1}{2} \| \vec{x}^{(t)} - \hatx^{(t-1)} \|^2$ (by $1$-strong convexity of $\cR_{\cX}$). Thus, an application of Cauchy-Schwarz inequality yields that 
\begin{equation}
    \label{eq:stab-x}
    \| \vec{x}^{(t)} - \hatx^{(t-1)} \|^2 \leq \eta \| \vec{x}^{(t)} - \hatx^{(t-1)} \| \|\mat{A} \vec{y}^{(t-1)} \|_* \implies \| \vec{x}^{(t)} - \hatx^{(t-1)} \| \leq \eta,
\end{equation}
since $\|\mat{A} \vec{y}^{(t-1)} \|_* \leq 1$ by the normalization assumption. Similar reasoning applied for the secondary sequence of \eqref{eq:OMD} implies that for any $t \in \N$,
\begin{equation}
    \label{eq:stab-x-}
    \| \hatx^{(t)} - \hatx^{(t-1)} \| \leq \eta.
\end{equation}

Now if $t = 1$, it follows from \eqref{eq:stab-x} that $\|\vec{x}^{(t)} - \vec{x}^{(t-1)}\| = \|\vec{x}^{(t)} - \hatx^{(t-1)}\| \leq \eta$ since $\vec{x}^{(0)} = \hatx^{(0)}$. Otherwise, for $t \geq 2$, applying the triangle inequality yields that $\|\vec{x}^{(t)} - \vec{x}^{(t-1)} \| \leq \| \vec{x}^{(t)} - \hatx^{(t-1)} \| + \|\vec{x}^{(t-1)} - \hatx^{(t-2)} \| + \| \hatx^{(t-1)} - \hatx^{(t-2)} \| \leq 3 \eta$ by \eqref{eq:stab-x} and \eqref{eq:stab-x-}. This completes the first part of the proof. Analogously, we conclude that $\| \vec{y}^{(t)} - \vec{y}^{(t-1)} \| \leq 3 \eta$ for any $t \in \N$.
\end{proof}

\begin{theorem}[Linear Decay of Regret; Full Version of \Cref{theorem:neg-reg}]
    \label{theorem:neg-reg-detailed}
    Suppose that both players in a bimatrix game $(\mat{A}, \mat{B})$ employ \eqref{eq:OMD} with $G$-smooth regularizer, learning rate $\eta > 0$ such that
    \begin{equation*}
    \eta \leq \min \left\{ \frac{1}{4\max\{\|\mat{A}\|_{\op}, \|\mat{B}\|_{\op} \}}, \frac{\epsilon^2}{96 \|\mat{A}\|_{\op} \|\mat{B}\|_{\op} \max \{\normset{\cX}, \normset{\cY} \} } \right\}
\end{equation*}
and 
\begin{equation*}
    T \geq \max \left\{ \frac{16 \max\{\normset{\cX}, \normset{\cY} \}}{\epsilon^2 \eta}, \frac{32 \max\{ \diamreg{\cR_{\cX}}, \diamreg{\cR_{\cY}}\}}{\epsilon^2\eta^2}, \frac{2048 \max\{ \diamreg{\cR_{\cY}} \normset{\cX}^2 \|\mat{A}\|^2_{\op}, \diamreg{\cR_{\cX}} \normset{\cY}^2 \|\mat{B}\|^2_{\op}\}}{\epsilon^4 \eta^2} \right\},
\end{equation*}
for some fixed $\epsilon > 0$. Then, if the dynamics do not reach a $(2\epsilon G \max\{\diamset{\cX}, \diamset{\cY}\} + \epsilon \eta)$-approximate NE, then 
    \begin{equation*}
        \max \{ \reg_{\cX}^T, \reg_{\cY}^T \} \leq - \min\left\{ \frac{\epsilon^2 \eta}{32}, \frac{\epsilon^4 \eta}{2048 \max\{ \normset{\cX}^2 \|\mat{A}\|^2_{\op}, \normset{\cY}^2 \|\mat{B}\|^2_{\op} \}} \right\} T.
    \end{equation*}
\end{theorem}

\begin{proof}
Suppose that there exists $t \in \range{T}$ such that 
\begin{equation*}
    \left( \| \vec{x}^{(t)} - \hatx^{(t)} \|^2 + \| \vec{x}^{(t)} - \hatx^{(t-1)}  \|^2 \right) + \left( \| \vec{y}^{(t)} - \haty^{(t)} \|^2 + \| \vec{y}^{(t)} - \haty^{(t-1)}  \|^2 \right) \leq \epsilon^2 \eta^2.
\end{equation*}

This would imply that $\|\vec{x}^{(t)} - \hatx^{(t)}\|, \| \vec{x}^{(t)} - \hatx^{(t-1)} \| \leq \epsilon \eta$ and $\|\vec{y}^{(t)} - \haty^{(t)}\|, \| \vec{y}^{(t)} - \haty^{(t-1)} \| \leq \epsilon \eta$. In turn, by \Cref{proposition:approx_stat} it follows that the pair of strategies $(\vec{x}^{(t)}, \vec{y}^{(t)})$ is a $(2 \epsilon G \max\{ \diamset{\cX}, \diamset{\cY} \} + \epsilon \eta)$-approximate Nash equilibrium, contradicting our assumption. As a result, we conclude that for all $t \in \range{T}$,
\begin{equation*}
    \left( \| \vec{x}^{(t)} - \hatx^{(t)} \|^2 + \| \vec{x}^{(t)} - \hatx^{(t-1)}  \|^2 \right) + \left( \| \vec{y}^{(t)} - \haty^{(t)} \|^2 + \| \vec{y}^{(t)} - \haty^{(t-1)}  \|^2 \right) \ge \epsilon^2 \eta^2.
\end{equation*}
Summing over all $t \in \range{T}$ yields that 
\begin{equation}
    \label{eq:linear_growth}
    \sum_{t=1}^T \left( \| \vec{x}^{(t)} - \hatx^{(t)} \|^2 + \| \vec{x}^{(t)} - \hatx^{(t-1)}  \|^2 \right) + \sum_{t=1}^T \left( \| \vec{y}^{(t)} - \haty^{(t)} \|^2 + \| \vec{y}^{(t)} - \haty^{(t-1)}  \|^2 \right) \geq \epsilon^2 \eta^2 T.
\end{equation}
We distinguish between two cases. First, we treat the case where 
\begin{equation}
    \label{eq:x-dom}
    \sum_{t=1}^T \left( \| \vec{x}^{(t)} - \hatx^{(t)} \|^2 + \| \vec{x}^{(t)} - \hatx^{(t-1)}  \|^2 \right) \ge \sum_{t=1}^T \left( \| \vec{y}^{(t)} - \haty^{(t)} \|^2 + \| \vec{y}^{(t)} - \haty^{(t-1)}  \|^2 \right).
\end{equation}
Then, by virtue of \eqref{eq:linear_growth},
\begin{equation}
    \label{eq:linear_growth-x}
    \sum_{t=1}^T \left( \| \vec{x}^{(t)} - \hatx^{(t)} \|^2 + \| \vec{x}^{(t)} - \hatx^{(t-1)}  \|^2 \right) \geq \frac{\epsilon^2 \eta^2}{2} T.
\end{equation}
Further, by the triangle inequality and Young's inequality,
\begin{equation*}
    \| \vec{y}^{(t)} - \vec{y}^{(t-1)} \|^2 \leq 2 \| \vec{y}^{(t)} - \haty^{(t-1)} \|^2 + 2 \| \haty^{(t-1)} - \vec{y}^{(t-1)} \|^2,
\end{equation*}
and summing over all $t \in \range{T}$ yields that
\begin{align*}
    \sum_{t=1}^T \| \vec{y}^{(t)} - \vec{y}^{(t-1)} \|^2 &\leq 2 \sum_{t=1}^T \| \vec{y}^{(t)} - \haty^{(t-1)} \|^2 + 2 \sum_{t=1}^T \| \haty^{(t-1)} - \vec{y}^{(t-1)} \|^2 \\
    &\leq 2 \sum_{t=1}^T \| \vec{y}^{(t)} - \haty^{(t-1)} \|^2 + 2 \sum_{t=1}^T \| \haty^{(t)} - \vec{y}^{(t)} \|^2,
\end{align*}
where the last inequality follows since $\haty^{(0)} = \vec{y}^{(0)}$. Hence, combining the latter bound with \eqref{eq:x-dom} implies that
\begin{equation}
    \label{eq:dom-x-new}
    \sum_{t=1}^T \left( \| \vec{x}^{(t)} - \hatx^{(t)} \|^2 + \| \vec{x}^{(t)} - \hatx^{(t-1)}  \|^2 \right) \geq \frac{1}{2} \sum_{t=1}^T \| \vec{y}^{(t)} - \haty^{(t-1)} \|^2.
\end{equation}
Now we are ready to bound the regret of player $\cX$. By \Cref{cor:rvu-our},
\begin{equation}
    \reg_{\cX}^T \leq \frac{\diamreg{\cR_{\cX}}}{\eta} + \eta \|\mat{A}\|^2_{\op} \sum_{t=1}^T \| \vec{y}^{(t)} - \vec{y}^{(t-1)} \|^2 - \frac{1}{4\eta} \sum_{t=1}^T \left( \| \vec{x}^{(t)} - \hatx^{(t)} \|^2 + \| \vec{x}^{(t)} - \hatx^{(t-1)} \|^2 \right).
\end{equation}
But \eqref{eq:dom-x-new} implies that 
\begin{align*}
    \eta \| \mat{A} \|^2_{\op} \sum_{t=1}^T \| \vec{y}^{(t)} - \vec{y}^{(t-1)} \|^2 - \frac{1}{8\eta} \sum_{t=1}^T &\left( \| \vec{x}^{(t)} - \hatx^{(t)} \|^2 + \| \vec{x}^{(t)} - \hatx^{(t-1)} \|^2 \right) \\ &\leq \left( \eta \| \mat{A} \|^2_{\op} - \frac{1}{16\eta} \right) \sum_{t=1}^T \| \vec{y}^{(t)} - \vec{y}^{(t-1)} \|^2 \le 0,
\end{align*}
since $\eta \le \frac{1}{4 \|\mat{A}\|_{\op}}$. From this we conclude that
\begin{equation}
    \reg_{\cX}^T \leq \frac{\diamreg{\cR_{\cX}}}{\eta} - \frac{1}{8\eta} \sum_{t=1}^T \left( \| \vec{x}^{(t)} - \hatx^{(t)} \|^2 + \| \vec{x}^{(t)} - \hatx^{(t-1)}  \|^2 \right) \leq \frac{\diamreg{\cR_{\cX}}}{\eta} - \frac{\epsilon^2 \eta}{16} T \leq - \frac{\epsilon^2 \eta}{32} T,
\end{equation}
for $T \geq \frac{32 \Omega_{\cR_{\cX}}}{\epsilon^2 \eta^2}$, where we used \eqref{eq:linear_growth-x}. Next, we focus on the regret of player $\cY$. By \eqref{eq:linear_growth-x} and \Cref{lemma:balanced}, 
\begin{equation}
    \sum_{t=1}^T \| \vec{y}^{(t)} - \vec{y}^{(t-1)} \| \geq \frac{\epsilon^2 \eta}{4 \normset{\cX} \| \mat{A} \|_{\op}} T - \frac{2}{\|\mat{A}\|_{\op}} \geq \frac{\epsilon^2 \eta}{8 \normset{\cX} \| \mat{A} \|_{\op}} T,
\end{equation}
since $T \geq \frac{16 \normset{\cX}}{\epsilon^2 \eta}$. Further, by Cauchy-Schwarz inequality, 
\begin{equation}
    \label{eq:linear_growth-y}
    \sum_{t=1}^T \| \vec{y}^{(t)} - \vec{y}^{(t-1)} \|^2 \geq \frac{1}{T} \left( \sum_{t=1}^T \| \vec{y}^{(t)} - \vec{y}^{(t-1)} \| \right)^2 \geq \frac{\epsilon^4 \eta^2}{64 \normset{\cX}^2 \|\mat{A}\|^2_{\op}} T.
\end{equation}
Now from \Cref{cor:rvu-our}, the regret of player $\cY$ can be bounded as
\begin{align*}
    \reg_{\cY}^T &\leq \frac{\diamreg{\cR_{\cY}}}{\eta} + \eta \|\mat{B}\|^2_{\op} \sum_{t=1}^T \| \vec{x}^{(t)} - \vec{x}^{(t-1)} \|^2 - \frac{1}{4\eta} \sum_{t=1}^T \left( \| \vec{y}^{(t)} - \haty^{(t)} \|^2 + \| \vec{y}^{(t)} - \haty^{(t-1)} \|^2 \right) \\
    &\leq \frac{\diamreg{\cR_{\cY}}}{\eta} + \eta \|\mat{B}\|^2_{\op} \sum_{t=1}^T \| \vec{x}^{(t)} - \vec{x}^{(t-1)} \|^2 - \frac{1}{8\eta} \sum_{t=1}^T \| \vec{y}^{(t)} - \vec{y}^{(t-1)} \|^2,
\end{align*}
where we used that 
\begin{align*}
    \sum_{t=1}^T \| \vec{y}^{(t)} - \vec{y}^{(t-1)} \|^2 \leq 2 \sum_{t=1}^T \| \vec{y}^{(t)} - \haty^{(t-1)} \|^2 + 2 \sum_{t=1}^T \| \haty^{(t)} - \vec{y}^{(t)} \|^2.
\end{align*}
Furthermore, by \Cref{lemma:stability} and \eqref{eq:linear_growth-y},
\begin{equation*}
    \eta \| \mat{B}\|^2_{\op} \sum_{t=1}^T \| \vec{x}^{(t)} - \vec{x}^{(t-1)} \|^2 \leq 9 \eta^3 \|\mat{B}\|^2_{\op} T \leq \frac{\epsilon^4 \eta}{1024 \normset{\cX}^2 \|\mat{A}\|^2_{\op}} T \le \frac{1}{16\eta} \sum_{t=1}^T \| \vec{y}^{(t)} - \vec{y}^{(t-1)} \|^2,
\end{equation*}
for $\eta \leq \epsilon^2 (96 \normset{\cX} \|\mat{A}\|_{\op} \|\mat{B}\|_{\op})^{-1}$. As a result, \eqref{eq:linear_growth-y} implies that
\begin{equation*}
    \reg_{\cY}^T \leq \frac{\diamreg{\cR_{\cY}}}{\eta} - \frac{1}{16\eta} \sum_{t=1}^T \|\vec{y}^{(t)} - \vec{y}^{(t-1)} \|^2 \leq \frac{\diamreg{\cR_{\cY}}}{\eta} - \frac{\epsilon^4 \eta}{1024 \normset{\cX}^2 \|\mat{A}\|^2_{\op}} T \leq - \frac{\epsilon^4 \eta}{2048 \normset{\cX}^2 \|\mat{A}\|^2_{\op}} T,
\end{equation*}
for $T \geq \frac{2048 \Omega_{\cR_{\cY}} \normset{\cX}^2 \|\mat{A}\|^2_{\op} }{\epsilon^4 \eta^2}$. Similarly, let us treat the case where
\begin{equation*}
    \sum_{t=1}^T \left( \| \vec{y}^{(t)} - \haty^{(t)} \|^2 + \| \vec{y}^{(t)} - \haty^{(t-1)}  \|^2 \right) \ge \sum_{t=1}^T \left( \| \vec{x}^{(t)} - \hatx^{(t)} \|^2 + \| \vec{x}^{(t)} - \hatx^{(t-1)}  \|^2 \right).
\end{equation*}
Then, for $\eta \leq \frac{1}{4 \|\mat{B}\|_{\op}}$ and $T \geq \frac{32 \diamreg{\cR_{\cY}} }{\epsilon^2 \eta^2}$,
\begin{equation*}
    \reg_{\cY}^T \leq \frac{\diamreg{\cR_{\cY}}}{\eta} - \frac{1}{8\eta} \sum_{t=1}^T \left( \| \vec{y}^{(t)} - \haty^{(t)} \|^2 + \| \vec{y}^{(t)} - \haty^{(t-1)} \|^2 \right) \leq \frac{\diamreg{\cR_{\cY}}}{\eta} - \frac{\epsilon^2 \eta}{16} T \leq - \frac{\epsilon^2 \eta}{32} T.
\end{equation*}
Moreover, for $T \ge \frac{16 \normset{\cY}}{\epsilon^2 \eta}$, 
\begin{equation*}
    \sum_{t=1}^T \| \vec{x}^{(t)} - \vec{x}^{(t-1)} \|^2 \ge \frac{\epsilon^4 \eta^2}{64 \normset{\cY}^2 \|\mat{B}\|_{\op}^2} T.
\end{equation*}
Thus, for $\eta \leq \epsilon^2 ( 96 \normset{\cY} \|\mat{A}\|_{\op} \|\mat{B}\|_{\op})^{-1} $,
\begin{equation*}
    \eta \|\mat{A}\|_{\op}^2 \sum_{t=1}^T \| \vec{y}^{(t)} - \vec{y}^{(t-1)} \|^2 \le 9 \eta^3 \|\mat{A}\|_{\op}^2 T \leq \frac{\epsilon^4 \eta}{1024 \normset{\cY}^2 \|\mat{B}\|_{\op}^2} T \leq \frac{1}{16\eta} \sum_{t=1}^T \| \vec{x}^{(t)} - \vec{x}^{(t-1)} \|^2.
\end{equation*}
Finally, for $T \geq \frac{2048 \diamreg{\cR_{\cX}} \normset{\cY}^2 \|\mat{B}\|^2_{\op}}{\epsilon^4 \eta^2}$,
\begin{equation*}
    \reg_{\cX}^T \leq \frac{\diamreg{\cR_{\cX}}}{\eta} - \frac{1}{16\eta} \sum_{t=1}^T \| \vec{x}^{(t)} - \vec{x}^{(t-1)} \|^2 \leq \frac{\diamreg{\cR_{\cX}}}{\eta} - \frac{\epsilon^4 \eta}{1024 \normset{\cY}^2 \|\mat{B}\|_{\op}^2} T \leq - \frac{\epsilon^4 \eta}{2048 \normset{\cY}^2 \|\mat{B}\|_{\op}^2} T.
\end{equation*}
\end{proof}

Next, we state the implication of \Cref{theorem:neg-reg-detailed} in normal-form games under \eqref{eq:OGD}. In that setting, it holds that $\normset{\cX}, \normset{\cY} = 1$; $\diamset{\cX}, \diamset{\cY} \leq \sqrt{2}$; $\diamreg{\cR_{\cX}}, \diamreg{\cR_{\cY}} \le 1$; and $G = 1$. Thus, we obtain the following simplified statement. 

\begin{corollary}[OGD in Normal-Form Games]
    \label{corollary:NFGs}
Suppose that both players in a bimatrix game $(\mat{A}, \mat{B})$ employ \eqref{eq:OGD} with learning rate $\eta > 0$ such that 
\begin{equation*}
    \eta \leq \min \left\{ \frac{1}{4\max\{\|\mat{A}\|_{\op}, \|\mat{B}\|_{\op} \}}, \frac{\epsilon^2}{96 \|\mat{A}\|_{\op} \|\mat{B}\|_{\op}} \right\}
\end{equation*}
and
\begin{equation*}
    T \geq \max \left\{ \frac{16}{\epsilon^2 \eta}, \frac{32}{\epsilon^2\eta^2}, \frac{2048 \max\{ \|\mat{A}\|^2_{\op}, \|\mat{B}\|^2_{\op}\}}{\epsilon^4 \eta^2} \right\},
\end{equation*}
for some fixed $\epsilon > 0$. Then,
\begin{itemize}
    \item Either there exists $t \in \range{T}$ such that the pair of strategies $(\vec{x}^{(t)}, \vec{y}^{(t)}) \in \cX \times \cY$ constitutes an $\epsilon(3 + \eta)$-approximate Nash equilibrium;
    \item Or, otherwise, the average correlated distribution of play after $T$ repetitions of the game is a 
    \begin{equation*}
        \min\left\{ \frac{\epsilon^2 \eta}{32}, \frac{\epsilon^4 \eta}{2048 \max\{ \|\mat{A}\|^2_{\op}, \|\mat{B}\|^2_{\op} \}} \right\}-\text{strong coarse correlated equilibrium.}
    \end{equation*}
\end{itemize}
\end{corollary}

Finally, we state an extension of \Cref{theorem:main-abridged} that establishes a dichotomy based on whether \emph{most} of the iterates are approximate Nash equilibria---not just a \emph{single} iterate. The proof is almost identical to the argument of~\Cref{theorem:neg-reg-detailed}, and is therefore omitted.

\begin{corollary}
    \label{cor:mostiter}
    Suppose that both players in a bimatrix game employ \eqref{eq:OGD} with learning rate $\eta = O(\epsilon^2 \delta)$ and $T = \Omega\left( \frac{1}{\eta^2 \epsilon^4 \delta^2} \right)$ repetitions, for a sufficiently small $\epsilon > 0$ and $\delta \in (0,1)$. Then,
    \begin{itemize}
        \item Either a $1 - \delta$ fraction of the iterates is an $\epsilon$-approximate Nash equilibrium;
        \item Or, otherwise, the average correlated distribution of play is an $\Omega(\epsilon^4 \eta \delta^2)$-strong CCE.
    \end{itemize}
\end{corollary}
\section{Description of the Game Instances}
\label{appendix:description}

In this section we provide a detailed description of the game instances we used in our experiments in \Cref{subsection:benchmark}. 

\paragraph{Liar's Dice} The first game we experimented on is \emph{Liar's dice}, a popular benchmark introduced by~\citet{Lisy15:Online}. In our instantiation, each of the two players initially privately roles a \emph{single} unbiased $4$-face die. Then, the first player announces any face value up to $4$, as well as the minimum number of dice the player believes have that value (among the dice of both players). Subsequently, each player in its own turn can either make a higher bid, or challenge the claim made by the previous player by declaring that player a ``liar''. In particular, a bid is higher than the previous one if either the face value is higher, or if the claimed number of dices is greater. In case the current player challenges the previous bid, all dice have to be revealed. If the claim was valid, the last bidder wins and receives a reward of $+1$, while the challenger incurs a negative payoff of $-1$. Otherwise, the utilities obtained are reversed. 

\paragraph{Sheriff} Our second benchmark is a bargaining game inspired by the board game \emph{Sheriff of Nottingham}, introduced by~\citet{Farina19:Correlation}. This game consists of two players: the \emph{smuggler} and the \emph{sheriff}. In our instantiation, the smuggler initially selects a number $n \in \{0, 1, 2, 3, 4, 5\}$ which corresponds to the number of \emph{illegal items} to be loaded in the cargo. Each illegal item has a fixed value of $1$. Next, $2$ rounds of bargaining between the two players follow. At each round, the smuggler decides on a \emph{bribe} ranging from $0$ to $b \defeq 3$ (inclusive), and the sheriff must decide whether or not the cargo will be inspected given the bribe amount. The sheriff's decision is binding only in the \emph{last round} of bargaining: if the sheriff accepts the bribe, the game stops with the smuggler obtaining a utility of $n$ minus the bribe amount $b$ proposed in the last bargaining round, while the sheriff receives a utility equal to $b$. In contrast, if the sheriff does not accept the bribe in last bargaining round and decides to inspect the cargo, there are two possible alternatives:
\begin{itemize}
    \item If the cargo has no illegal items (\emph{i.e.} $n = 0$), the smuggler receives the fixed amount of $3$, while sheriff incurs a negative payoff of $-3$;
    \item Otherwise, the utility of the smuggler is set to $-2n$, while the utility of the Sheriff is $2n$.
\end{itemize}

\paragraph{Battleship} Our next benchmark is \emph{Battleship}, a parametric version of the popular board game introduced in~\citep{Farina19:Correlation}. At the beginning, each player secretly places its ships on separate locations on a grid of size $2 \times 2$. Every ship has size $1$ and a value of $4$, and the placement is such that there is no overlap with any other ship. After the placement, players take turns at ``firing'' at their opponent's ships. The game proceeds until either one player has sunk all of the opponent's ships, or each player has completed $r = 2$ rounds of firing. At the end of the game, each player's payoff is the sum of the values of the opponent's ships that were sunk, minus the sum of the values of the ships that the player has lost \emph{multiplied by two}. The latter modification makes the game general-sum, and incentivizes players to be more risk-averse.

\paragraph{Goofspiel} Our final benchmark is \emph{Goofspiel}, introduced by~\citet{Ross71:Goofspiel}. In this game every player has a hand of cards numbered from $1$ to $h$, where in our instantiation $h \defeq 3$. An additional stack of $h$ cards is shuffled and singled out as winning the current prize. In every turn a prize card is revealed, and players privately choose one of their cards to bid. The player with the highest card wins the current prize, while in case of a tie the prize card is discarded. Due to this tie-breaking mechanism, even two-player instances are general-sum. After the completion of $h$ turns, players obtain the sum of the values of the prize cards they have won. Further, the instances we consider are of \emph{limited information}---the actions of the other player are observed only at the end of the game. This makes the game strategically more involved as each player has less information about the opponent's actions. 

\end{document}